\documentclass[10pt,twoside]{article}
\usepackage{amsmath,amsthm}
\usepackage{amssymb}
\usepackage[latin1]{inputenc}
\usepackage[T1]{fontenc}
\usepackage{bbm}
\def\YEAR{\year}\newcount\VOL\VOL=\YEAR\advance\VOL by-1995
\def\firstpage{1}\def\lastpage{5}
\def\received{}\def\revised{}
\def\communicated{}

\makeatletter
\def\magnification{\afterassignment\m@g\count@}
\def\m@g{\mag=\count@\hsize6.5truein\vsize8.9truein\dimen\footins8truein}
\makeatother

\font\eightrm=cmr8
\font\caps=cmcsc10                    
\font\Caps=cmcsc10 scaled \magstep1   

\makeatletter
\@specialpagefalse\headheight=8.5pt
\def\DocMath{}
\renewcommand{\@evenhead}{%
    \ifnum\thepage>\lastpage\rlap{\thepage}\hfill%
    \else\rlap{\thepage}\slshape\leftmark\hfill{\caps\SAuthor}\hfill\fi}%
\renewcommand{\@oddhead}{%
    \ifnum\thepage=\firstpage{\DocMath\hfill\llap{\thepage}}%
    \else{\slshape\rightmark}\hfill{\caps\STitle}\hfill\llap{\thepage}\fi}%
\makeatother

\def\TSkip{\bigskip}
\newbox\TheTitle{\obeylines\gdef\GetTitle #1
\ShortTitle  #2
\SubTitle    #3
\Author      #4
\ShortAuthor #5
\EndTitle
{\setbox\TheTitle=\vbox{\baselineskip=20pt\let\par=\cr\obeylines%
\halign{\centerline{\Caps##}\cr\noalign{\medskip}\cr#1\cr}}%
	\copy\TheTitle\TSkip\TSkip%
\def\next{#2}\ifx\next\empty\gdef\STitle{#1}\else\gdef\STitle{#2}\fi%
\def\next{#3}\ifx\next\empty%
    \else\setbox\TheTitle=\vbox{\baselineskip=20pt\let\par=\cr\obeylines%
    \halign{\centerline{\caps##} #3\cr}}\copy\TheTitle\TSkip\TSkip\fi%
\centerline{\caps #4}\TSkip\TSkip%
\def\next{#5}\ifx\next\empty\gdef\SAuthor{#4}\else\gdef\SAuthor{#5}\fi%
\ifx\received\empty\relax
    \else\centerline{\eightrm Received: \received}\fi%
\ifx\revised\empty\TSkip%
    \else\centerline{\eightrm Revised: \revised}\TSkip\fi%
\ifx\communicated\empty\relax
    \else\centerline{\eightrm Communicated by \communicated}\fi\TSkip\TSkip%
\catcode'015=5}}\def\Title{\obeylines\GetTitle}
\def\Abstract{\begingroup\narrower
    \parskip=\medskipamount\parindent=0pt{\caps Abstract. }}
\def\EndAbstract{\par\endgroup\TSkip}

\long\def\MSC#1\EndMSC{\def\arg{#1}\ifx\arg\empty\relax\else
     {\par\narrower\noindent%
     2000 Mathematics Subject Classification: #1\par}\fi}

\long\def\KEY#1\EndKEY{\def\arg{#1}\ifx\arg\empty\relax\else
	{\par\narrower\noindent Keywords and Phrases: #1\par}\fi\TSkip}

\newbox\TheAdd\def\Addresses{\vfill\copy\TheAdd\vfill
    \ifodd\number\lastpage\vfill\eject\phantom{.}\vfill\eject\fi}
{\obeylines\gdef\GetAddress #1
\Address #2
\Address #3
\Address #4
\EndAddress
{\def\xs{4.3truecm}\parindent=0pt
\setbox0=\vtop{{\obeylines\hsize=\xs#1\par}}\def\next{#2}
\ifx\next\empty 
     \setbox\TheAdd=\hbox to\hsize{\hfill\copy0\hfill}
\else\setbox1=\vtop{{\obeylines\hsize=\xs#2\par}}\def\next{#3}
\ifx\next\empty 
     \setbox\TheAdd=\hbox to\hsize{\hfill\copy0\hfill\copy1\hfill}
\else\setbox2=\vtop{{\obeylines\hsize=\xs#3\par}}\def\next{#4}
\ifx\next\empty\ 
     \setbox\TheAdd=\vtop{\hbox to\hsize{\hfill\copy0\hfill\copy1\hfill}
                \vskip20pt\hbox to\hsize{\hfill\copy2\hfill}}
\else\setbox3=\vtop{{\obeylines\hsize=\xs#4\par}}
     \setbox\TheAdd=\vtop{\hbox to\hsize{\hfill\copy0\hfill\copy1\hfill}
	        \vskip20pt\hbox to\hsize{\hfill\copy2\hfill\copy3\hfill}}
\fi\fi\fi\catcode'015=5}}\gdef\Address{\obeylines\GetAddress}

\newcommand{\R}{\mathbbm{R}}
\newcommand{\N}{\mathbbm{N}}

\newcommand{\C}{\mathbbm{C}}

\newtheorem{Satz}{Theorem}[section]
\newtheorem{Def}[Satz]{Definition}

\newtheorem{Lemma}[Satz]{Lemma}

\newtheorem{Bem}[Satz]{Remark}
\newtheorem{Hyp}{Hypothesis}
\newcommand{\tensor}{\otimes}

\newcommand{\ran}{\operatorname{ran}}

\newcommand{\dom}{\operatorname{dom}}
\newcommand{\Fb}{\mathcal{F}_b}
\newcommand{\Tr}{\operatorname{Tr}}

\newcommand{\cl}{\operatorname{cl}}
\newcommand{\one}{\mathbf{1}}
\newcommand{\unl}{\underline}
\newcommand{\const}{\operatorname{const}}
\newcommand{\Hel}{\mathcal{H}_{el}}

\newcommand{\Lel}{\mathcal{L}_{el}}
\newcommand{\Hael}{H_{el}}
\newcommand{\Haelp}{H_{el,+}}

\newcommand{\piel}{\pi^{el}}
\newcommand{\omel}{\omega^\beta_{el}}
\newcommand{\Omel}{\Omega^\beta_{el}}
\newcommand{\Zel}{\mathcal{Z}}

\newcommand{\We}{W}

\newcommand{\Lelaux}{ \mathcal{L}_{el,a}}
\newcommand{\Wop}{ \mathcal{W}}
\newcommand{\omf}{ \omega^\beta_f}

\newcommand{\pif}{ \pi_f}

\newcommand{\Haf}{ \check{H} }

\newcommand{\Omf}{\Omega^\beta_{f}}
\newcommand{\Lf}{ \mathcal{L}_f}
\newcommand{\Lfaux}{ \mathcal{L}_{f,a}}

\newcommand{\hph}{\mathcal{H}_{ph}}
\newcommand{\Kg}{\mathcal{K}}
\newcommand{\Hosc}{H_{osc}}

\newcommand{\be}{A}
\newcommand{\omosc}{\omega_\beta^{osc}}
\newcommand{\pig}{\pi}
\newcommand{\hc}{\operatorname{h.c.}}
\newcommand{\Ag}{\mathfrak{A}}
\newcommand{\Qg}{Q}
\newcommand{\QN}{Q_N}
\newcommand{\Jg}{\mathcal{J}}
\newcommand{\Lg}{\mathcal{L}_Q}
\newcommand{\Laux}{\mathcal{L}_{a}}
\newcommand{\Lo}{\mathcal{L}_0}
\newcommand{\tauo}{ \tau^0}
\newcommand{\taug}{\tau}
\newcommand{\taun}{\tau^N}
\newcommand{\Hg}{H}
\newcommand{\Wg}{W}
\newcommand{\Hig}{\mathcal{H}}
\newcommand{\Ho}{H_0}
\newcommand{\f}{\mathfrak{f}}

\newcommand{\Omo}{\Omega_0^\beta}
\newcommand{\omo}{\omega_0^\beta}
\newcommand{\Omg}{\Omega^\beta}
\newcommand{\omg}{\omega^\beta}
\newcommand{\Mg}{\mathfrak{M}_\beta}
\newcommand{\Mana}{\mathfrak{M}_\beta^{ana}}
\newcommand{\slim}{\operatorname{s-lim}}
\newcommand{\wlim}{\operatorname{w-lim}}
\renewcommand{\leq}{\leqslant}
\renewcommand{\le}{\leqslant}
\renewcommand{\geq}{\geqslant}
\renewcommand{\ge}{\geqslant}
\renewcommand{\Re}{\textrm{Re}\,}
\renewcommand{\Im}{\textrm{Im}\,}
\renewcommand{\imath}{i}
\newcommand{\Core}{\mathcal{D}}
\newcommand{\linhull}{\operatorname{LH}}
\newcommand{\ovl}{\overline}
\begin{document}
\Title
		An Infinite Level Atom coupled to a Heat Bath
\ShortTitle
			An Infinite Level Atom
\SubTitle
\Author
		Martin Koenenberg%
	\footnote{\eightrm  Supported by the DFG (SFB/TR 12)} %
	
\ShortAuthor
		M.Koenenberg   
\EndTitle

\Abstract
We consider a $W^*$-dynamical system $(\Mg,\taug)$, which models
finitely many particles coupled to an infinitely extended heat
bath. The energy of the particles can be described by an unbounded operator,
which has infinitely many energy levels.
We show existence of the dynamics $\taug$ and existence of a $(\beta,\taug)$
-KMS state under very explicit conditions on the strength of the interaction
and on the inverse temperature $\beta$.
\EndAbstract

\MSC
Primary 81V10; Secondary 37N20, 47N50.
\EndMSC

\KEY
KMS-state, thermal equilibrium, $W^*$-dynamical system, open quantum system.
\EndKEY

\Address
Martin K\"onenberg
Fakult\"at f\"ur Mathematik und Informatik
FernUniversit\"at Hagen
L\"utzowstra{\ss}e 125
D-58084 Hagen, Germany.
martin.koenenberg@fernuni
-hagen.de
\Address
\Address
\Address
\EndAddress
%
%
%

%
%
%
%
\section{Introduction}
In this paper, we study a $W^*$-dynamical system $(\Mg,\taug)$
which describes a system of finitely many particles interacting
with an infinitely extended bosonic reservoir or heat bath 
at inverse temperature $\beta$.
Here, $\Mg$ denotes the $W^*$-algebra of observables and $\taug$
is an automorphism-group on $\Mg$, which is defined by
\begin{equation}
\taug_t(X):= e^{\imath t \Lg}\,X\, e^{-\imath t \Lg},\ X\in \Mg,\ t\in \R.
\end{equation}
In this context, $t$ is the time parameter. $\Lg$ is the Liouvillean
of the dynamical system at inverse temperature $\beta$, $\Qg$ describes 
the interaction between particles and heat bath. On the one hand
the choice of $\Lg$ is motivated by heuristic arguments, which allow
to derive the Liouvillean $\Lg$ from the Hamiltonian $\Hg$ of the joint system
of particles and bosons at temperature zero. On the other hand we ensure that $\Lg$ anti-commutes
with a certain anti-linear conjugation $\Jg$, that will be introduced later on.
The Hamiltonian, which represents the interaction with a bosonic gas at 
temperature zero, can be  the Standard Hamiltonian of the
non-relativistic QED, (see or instance \cite{BachFroehlichSigal1998a}), or the Pauli-Fierz operator,
which is defined in \cite{DerezinskiJaksic2001,BachFroehlichSigal1998a}, or the Hamiltonian of Nelson's Model.
We give the definition of these Hamiltonians in the sequel of Definition \ref{Def:Ho-Hg}.\\
\indent Our first result is the following:
\begin{Satz}
$\Lg$, defined in \eqref{Def:Lg}, has a unique self-adjoint realization and
$\taug_t(X)\in \Mg$ for all $t\in \R$ and all $X\in \Mg$.
\end{Satz}
The proof follows from Theorem \ref{Thm1} and Lemma \ref{LemInv}.
The main difficulty in the proof is,  that $\Lg$ is not semi-bounded,
and that one has to define a suitable auxiliary operator in order to
apply Nelson's commutator theorem.\\
\indent Partly, we assume that the isolated system of
finitely many particles is confined in space. This is reflected in Hypothesis \ref{Hyp1},
where we assume that the particle Hamiltonian $\Hael$ possesses a Gibbs state.
In the case where $\Hael$ is a Schrödinger-operator, we give in Remark 
\ref{Bem:Gibbs} a sufficient condition on the external potential $V$
to ensure the existence of a Gibbs state for $\Hael$. Our second theorem
is
\begin{Satz}
Assume Hypothesis \ref{Hyp1} and that $\Omo\in \dom(e^{-(\beta/2)(\Lo+\Qg)})$. 
Then there exists a $(\beta,\taug)$-KMS state $\omg$ on $\Mg$.
\end{Satz}
This theorem ensures the existence of an equilibrium state on $\Mg$ for
the dynamical system $(\Mg,\taug)$. Its proof is part of Theorem
\ref{Thm3} below. Here, $\Lo$ denotes the Liouvillean for the joint system of particles
and bosons, where the interaction part is omitted. 
$\Omo$ is the vector representative of the $(\beta,\taug)$-KMS state
for the system without interaction. In a third theorem we study the
condition $\Omo\in \dom(e^{-(\beta/2)(\Lo+\Qg)})$:
\begin{Satz}\label{Thm2}
Assume Hypothesis \ref{Hyp1} is fulfilled.
Then there are two cases,
\begin{enumerate}
\item If $0\,\leq\, \gamma\, <\,1/2$ and $\unl{\eta}_1\,(1+\beta)\,\ll\, 1$, then 
$\Omo\in \dom(e^{-\beta/2\,(\Lo+\Qg)})$.
\item If $\gamma=1/2$ and $(1+\beta)(\unl{\eta}_1+\unl{\eta}_2)\,\ll\, 1$, then
           $\Omo\in \dom(e^{-\beta/2\,(\Lo+\Qg)})$.
\end{enumerate}
\end{Satz}
Here, $\gamma\in[0,1/2)$ is a parameter of the model, see \eqref{FormFactorBound} and
$\unl{\eta}_1,\unl{\eta}_2$ are parameters, which describe the strength of the
interaction, see \eqref{FormFactorBound}. In a last theorem we consider the case
where $\Hael= -\Delta_q + \Theta^2q^2$ and the interaction Hamiltonian
is $\lambda\, q\, \Phi(f)$ at temperature zero for $\lambda \not=0$. Then,
\begin{Satz} \label{Thm4a}
$\Omo$ is in $\dom(e^{-\beta/2\,(\Lo\,+\,\Qg)})$ for all
$\beta \in (0,\, \infty)$,
whenever $$|2\Theta^{-1}\,\lambda|\,\||k|^{-1/2}\,f\|_{\hph}\, <\,1.$$
\end{Satz}
Furthermore, we
show that our strategy  can not be improved to obtain a result, which
ensures existence for all values of $\lambda$, see \eqref{eq4.12}.\\
\indent In the last decade there appeared a large number of mathematical
contributions to the theory of open quantum system. Here we only want to
mention some of them \cite{BachFröhlichSigal2000,DerezinskiJaksicPillet2003,DerezinskiJaksic2003,FroehlichMerkli2004,FroehlichMerkliSigal2004,JaksicPillet1996a,JaksicPillet1996b,Merkli2001}, which consider a
related model, in which the particle Hamilton $\Hael$ is represented as a finite symmetric matrix
and the interaction part of the Hamiltonian is linear in annihilation  and creation operators.
In this case one can prove existence of a $\beta$- KMS without
any restriction to the strength of the coupling. (In this case we can apply Theorem
\ref{Thm2} with $\gamma=0$ and $\unl{\eta}_1=0$). 
We can show existence of KMS-states for an infinite level atom coupled to a
heat bath.
Furthermore, in \cite{DerezinskiJaksicPillet2003}
there is a general theorem, which ensures existence of a $(\beta,\taug)$-KMS
state under the assumption, that $\Omo\in \dom(e^{-(\beta/2) \Qg})$, which implies
$\Omo\in \dom(e^{-(\beta/2)(\Lo+\Qg)})$. In Remark \ref{Bem3d} we verify that
this condition implies the existence of a $(\beta,\taug)$-KMS state in the case of
a harmonic oscillator with dipole interaction $\lambda\, q\, \cdot \,\Phi(f)$,
whenever $(1+\beta) \lambda\|(1+|k|^{-1/2})f\| \ll 1$.\\
\bigskip
%
%
%
\section{Mathematical Preliminaries}
\subsection{Fock Space, Field- Operators and Second Quantization}
We start our mathematical introduction with the description of the
joint system of particles and bosons at temperature zero. 
The Hilbert space describing bosons at temperature zero is
the \emph{bosonic Fock space} $\Fb$, where
\begin{equation*} \label{FockSpace}
\Fb\,:=\, \mathcal{F}_b[\hph] \,:=\, \C \oplus \bigoplus_{n=1}^\infty \hph^{(n)},
\qquad \hph^{(n)}\,:=\,\bigotimes_{sym}^n \hph.
\end{equation*}
$\hph$ is  either  a closed subspace of $L^2(\R^3)$ or $L^2(\R^3 \times \{\pm\})$,
being invariant under complex conjugation.
If phonons are considered we choose $\hph= L^2(\R^3 )$,
if photons are considered we choose $\hph=L^2(\R^3 \times \{\pm\})$.
In the latter case "+" or "-" labels the polarization of the photon. 
However, we will write $\langle\, f\,|\,g\,\rangle_{\hph} \,:=\, \int \ovl{f(k)}\, g(k)\,dk$ 
for the scalar product in both cases. This is an abbreviation for
$\sum_{p\,=\,\pm} \int \ovl{f(k,p)}\, g(k,p)\,dk$
in the case of photons.\\
\indent $\hph^{(n)}$ is the $n$-fold symmetric tensor product of $\hph$, that is, it contains
all square integrable functions $f_n$ being invariant under permutations $\pi$ of
the variables, i.e., $f_n(k_1,\ldots,k_n)\,=\,f_n(k_{\pi(1)},\ldots,k_{\pi(n)})$. 
For phonons we have $k_j\in \R^3$ and
 $k_j\in \R^3\times \{\pm\}$ for photons. 
The wave functions in $\hph^n$ are states of $n$ bosons.\\
\indent The vector $\Omega \,:=\, (1,\,0,\,\ldots)\,\in \Fb$ is called the \emph{vacuum}. 
Furthermore we denote the subspace $\Fb$ of finite sequences with ${\Fb}^{fin}$. 
On ${\Fb}^{fin}$ the \emph{creation and annihilation operators},
$a^*(h)$ and $a(h)$, are defined for $h\in \hph$ by
\indent 
\begin{align}\label{CreaAnnOps}
(a^*(h)&\,f_{n})(k_1,\ldots,k_{n+1})\\ \nonumber
&\,=\,(n+1)^{-1/2}\sum_{i=1}^{n+1}h(k_i)\,f_{n}(k_1,\ldots,k_{i-1},k_{i+1},\ldots ,k_{n+1}),\\ 
(a(h)&\,f_{n+1})(k_1,\ldots,k_{n})\\ \nonumber
&\,=\,(n+1)^{1/2}\int h(k_{n+1})\,f_{n+1}(k_1,\ldots,k_{n+1})\,dk_{n+1},\nonumber
\end{align}
and $a^*(h)\,\Omega\,=\,h,\ a(h)\,\Omega\,=\,0 $. Since $a^*(h)\,\subset\,(a(h))^*$ and $a(h)\,\subset\,(a^*(h))^*$,
the operators $a^*(h)$ and $a(h)$ are closable. Moreover, the  canonical commutation relations (CCR) hold
true, i.e.,
\begin{gather*}\label{CCR}
[a(h)\,,\,a(\widetilde{h})]\,=\,[a^*(h)\,,\,a^*(\widetilde{h})]\,=\,0,
\qquad [a(h)\,,\,a^*(\widetilde{h})] \,=\,\langle\, h \,|\, \widetilde{h}\,\rangle_{\hph}.
\end{gather*}
Furthermore we define field operator by
\begin{equation*} \label{FieldOp}
\Phi(h)\,:=\,2^{-1/2}\,(a(h)\,+\,a^*(h)),\qquad h\in \hph.
\end{equation*}
It is a straightforward calculation to check that the vectors in
 ${\Fb}^{fin}$ are analytic for $\Phi(h)$. Thus, $\Phi(h)$ is 
essentially self-adjoint on ${\Fb}^{fin}$. 
In the sequel, we will identify $a^*(h),\, a(h)$ and $ \Phi(h)$
with their closures. The Weyl operators $\We(h)$ are given by
$\We(h)\,=\,\exp(\imath\, \Phi(h))$. 
They fulfill the CCR-relation for the Weyl operators, i.e.,
\begin{equation*} \label{WeylCCR}
\We(h)\,\We(g)\,=\, \exp( \imath /2\,  \Im\langle\, h\,|\,g\,\rangle_{\hph} )\,\We(g+h),
\end{equation*}
which follows from explicit calculations on ${\Fb}^{fin}$. The
Weyl algebra $\We(\f)$ over  a subspace $\f$ of $\hph$ is defined by
\begin{equation}\label{WeylAlgebra}
\We(\f) \,:= \,\cl \linhull \{ \We(g)\in \mathcal{B}(\Fb)\,:\, g\in \f\}.
\end{equation}
Here, $\cl$ denotes the closure with respect to the norm of $\mathcal{B}(\Fb)$,
and "$\linhull$" denotes the linear hull.\\
\indent Let $\alpha\,:\, \R^3\,\rightarrow\, [0,\,\infty)$ be a locally bounded Borel function
and $\dom(\alpha )\,:=\, \{ f\in \hph\,:\, \alpha f\in \hph \}$. Note, that $(\alpha f)(k)$ 
is given by $\alpha(k)\,f(k,p)$ for photons.
If $\dom(\alpha)$ is dense subspace of $\hph$, $\alpha$ defines a self-adjoint multiplication
operator on $\hph$. In this case, the second quantization $d\Gamma(\alpha)$ of $\alpha$  is defined by
\begin{equation*}\label{SecQuant}
(d\Gamma(\alpha)\,f_n)(k_1,\ldots,k_n)\,:=\, (\alpha(k_1)+\alpha(k_2)+\ldots+\alpha(k_n))\,f_n(k_1,\ldots,k_n) 
\end{equation*}
and $d\Gamma(\alpha)\,\Omega\,=\,0$ on its maximal domain.
%
%
\subsection{Hilbert space and Hamiltonian for the particles}
%
%
Let $\Hel$ be a closed, separable subspace of $L^2(X,\, d\mu)$, that is 
invariant under complex conjugation. 
The Hamiltonian $\Hael$ for the particle is a self-adjoint operator on $\Hel$
being bounded from below. We set $\Haelp\,:=\,\Hael\,-\,\inf\sigma(\Hael)\,+\,1$.
Partly, we need the assumption
\begin{Hyp}\label{Hyp1}
Let $\beta>0$. There exists a small positive constant $\epsilon>0$, and
\begin{equation*}\label{GibbsCond}
\Tr_{\Hel}\{ e^{-(\beta\,-\,\epsilon)\, \Hael}\}\, <\,\infty .
\end{equation*}
The condition implies the existence of a Gibbs state
\begin{equation*}\label{GibbsState}
\omel(A)\,=\,\Zel^{-1}\,\Tr_{\Hel}\{ e^{-\beta\, \Hael} A \},\quad A\in \mathcal{B}(\Hel),
\end{equation*}
for $\Zel \, = \,\Tr_{\Hel}\{ e^{-\beta\, \Hael} \}$.
\end{Hyp}

\begin{Bem} \label{Bem:Gibbs}
Let  $\Hel=L^2(\R^n, d^n\,x)$ 
and $\Hael\,=\, -\Delta_x \,+\, V_1\,+\,V_2$,  
where $V_1$ is a $-\Delta_x$-bounded potential
with relative bound $a\,<\,1$ and $V_2$ is in $L^2_{loc}(\R^n,\, d^nx)$.
Thus $\Hael$ is essentially self-adjoint on $\mathcal{C}_c^\infty(\R^n)$.
Moreover, if additionally 
\begin{equation}
\int e^{-(\beta\,-\,\epsilon)\, V_2(x)}\,d^n\,x<\,\infty,
\end{equation}
 then one can show, using the Golden-Thompson-inequality, that
Hypothesis \ref{Hyp1} is satisfied.
\end{Bem}
\subsection{Hilbert space and Hamiltonian 
for the interacting system}\label{SecGeneralizedFieldop}
The Hilbert space for the joint system is $\Hig\,:=\, \Hel \tensor \Fb$.
The vectors in $\Hig$ are sequences $f= (f_n)_{n\in \N_0}$ of wave functions,
$f_n\in \Hel \tensor \hph^{(n)}$, obeying
\begin{gather*}
\unl{k}_n \,\mapsto \,f_n(x,\,\unl{k}_n)\in  \hph^{(n)} 
\qquad \textrm{ for } \mu\textrm{- almost every x}\\
x\, \mapsto \,f_n(x,\,\unl{k}_n)\in  \Hel\qquad  
\text{ for Lebesgue - almost every }\unl{k}_n,
\end{gather*}
where $\unl{k}_n \,=\, (k_1,\ldots,k_n)$. 
The complex conjugate vector is $\ovl{f}\,:=\, (\, \ovl{f}_n\,)_{n\in \N_0}$.\\
\indent Let $G^j\,:=\,\{G^j_k\}_{k\in\R^3},\, H^j\,:=\,\{H^j_k\}_{k\in\R^3}$ and $F\,:=\,\{F_k\}_{k\in\R^3}$ 
be families of closed operators on $\Hel$ for $j\,=\,1,\ldots,r$.
We assume, that $\dom(F^\sigma_k)\,\supset\, \dom(\Haelp^{1/2})$ and that
%
$$
 k\,\mapsto\, G_k^j,\ (H_k^j),\ F_k \Haelp^{-1/2},\ (F_k)^* \Haelp^{-1/2}   \in \mathcal{B}(\Hel)
$$
%
are weakly (Lebesgue-)measurable. For $\phi \in \dom(\Haelp^{1/2})$ we assume that
\begin{gather}
k\,\mapsto\, (G_k^j\,\phi)(x),\ (H_k^j\,\phi)(x),\ (F_k\,\phi)(x)   \in \hph,\\
k\,\mapsto\, ((G_k^j)^*\,\phi)(x),\ ((H_k^j)^*\,\phi)(x),\ ((F_k)^*\,\phi)(x)   \in \hph, \textrm{ for }x\in X.
\end{gather}
Moreover we assume for $\vec G\,=\,(G^1,\ldots,\,G^r),\ 
\vec H \,:=\, (H^1,\ldots,\,H^r)$ and $F$, that
\begin{equation*}
\|\vec G\|_w\,<\,\infty,\ \|\vec H\|_w \,<\,\infty ,\  \|F\|_{w,1/2}\,<\,\infty,
\end{equation*}
where 
\begin{gather*}\label{DefNorm}
\|G_j\|^2_w \,:=\,\int (\alpha(k)\,+\,\alpha(k)^{-1})
\,\big(\|(G_k^j)^*\|_{\mathcal{B}(\Hael)}^2\,+\,\|G_k^j\|_{\mathcal{B}(\Hael)}^2\big)\,dk\\
\|\vec G\|^2_w\, :=\, \sum_{j=1}^r\|G_j\|^2_w,
\qquad \|F\|^2_{w,1/2}\,:=\, \|F\Haelp^{-1/2}\|^2_{w}\,+\,\|F^*\Haelp^{-1/2}\|^2_{w}.
\end{gather*}
\noindent We define for $f\,=\,(f_n)_{n=0}^\infty\in \dom(\Haelp^{1/2})\tensor \Fb^{fin}$ the (generalized) creation operator
\begin{align}\label{GenCreaOp}
(a^*(F)\,f_{n})&(x,\,k_1,\ldots,\,k_{n+1})\\ \nonumber
 &\,:=\,(n+1)^{-1/2}\,\sum_{i=1}^{n+1}(F_{k_{i}}\,f_{n})(x,\,k_1,\ldots,\,k_{i-1},\,k_{i+1},\ldots ,\,k_{n+1})
\end{align}
and $a(F)\,f_0(x)\,=\, 0$. The (generalized) annihilation operator is
\begin{align}\label{GenAnnOp}
(a(F)\,f_{n+1})&(x,\,k_1,\ldots,k_{n})\\ \nonumber
&\,:=\,(n+1)^{1/2}\,\int(F^*_{k_{n+1}}\,f_{n+1})(x,\,k_1,\ldots,\,k_{n},\,k_{n+1})\, dk_{n+1}.
\end{align}
Moreover, the corresponding (generalized) field operator is $ \Phi(F)\,:=\, 2^{-1/2}\, (a(F)\,+\,a^*(F))$.
$\Phi(F)$ is symmetric on $\dom(\Haelp^{1/2})\tensor \Fb^{fin}$.
The bounds follow directly from Equations \eqref{GenCreaOp} and \eqref{GenAnnOp}.
\begin{eqnarray}\label{RelBoundGen}
\|\,a(F)\Haelp^{-1/2}\,f\|^2_{\Hig}
&\,\leq\,& \int |\alpha(k)|^{-1} \|F^*_k\,\Haelp^{-1/2}\|^2_{\mathcal{B}(\Hel)}\,dk \cdot
\|d\Gamma(|\alpha|)^{1/2}f\|^2_{\Hig}\\ \nonumber
\|a^*(F)\Haelp^{-1/2}\,f\|^2_{\Hig}
&\,\leq\,& \int |\alpha(k)|^{-1} \|F_k\,\Haelp^{-1/2}\|^2_{\mathcal{B}(\Hel)}\,dk\cdot
\|d\Gamma(|\alpha|)^{1/2}\,f\|^2_{\Hig}\\  \nonumber
&&\, +\,\int \|F_k\,\Haelp^{-1/2}\|^2_{\mathcal{B}(\Hel)}\,dk
\cdot \|f\|^2_{\Hig}.
\end{eqnarray}
For $(G_k)^j,\, (H_k)^j\in \mathcal{B}(\Hel)$, the factor $\Haelp^{-1/2}$ can
be omitted. The Hamiltonians for the non-interacting, resp. interacting model are
\begin{Def} \label{Def:Ho-Hg}On $\dom(\Hael)\tensor \dom(d\Gamma(\alpha))\cap \Fb^{fin}$ we define
\begin{equation}\label{Def:Ho-Hg}
\Ho\, :=\, \Hael \tensor \one \,+\, \one \tensor d\Gamma(\alpha), \qquad
\Hg\, :=\, \Ho \,+\, \Wg,
\end{equation}
where $\Wg\, :=\, \Phi(\vec G )\,\Phi(\vec H )\,+\, \hc\,+\,\Phi(F)$
and $\Phi(\vec G )\,\Phi(\vec H )\,:=\, \sum_{j=1}^r\Phi( G^j )\,\Phi(H^j )$.
The abbreviation "h.c." means the formal adjoint operator of $\Phi(\vec G )\,\Phi(\vec H )$.
\end{Def}
We give  examples for possible configurations:\\
Let $\gamma \in \R$ be a small coupling parameter.\smallskip\\
$\blacktriangleright$
The Nelson Model:\\
$\Hel\subset L^2(\R^{3N})$, $\Hael\,:=\, -\Delta\,+\,V$,\, $\hph\,=\,L^2(\R^3)$ and $\alpha(k)\,=\,|k|$.
The form factor is $F_k\,=\,\gamma\, \sum_{\nu\,=\,1}^N\,e^{- \imath\, k x_\nu}\,|k|^{-1/2}\, 
\one[\,|k|\,\leq\, \kappa],\ x_\nu\in \R^3$ and $H^j,\, G^j\,=\,0$.\smallskip\\
$\blacktriangleright$ The Standard Model of Nonrelativistic  QED:\\
$\Hel\subset L^2(\R^{3N})$, $\Hael\,:=\, -\Delta\,+\,V$,\, $\hph\,=\,L^2(\R^3\times \{\pm\})$ and $\alpha(k)\,=\,|k|$.
The form factors are 
\begin{gather*}
F_\mathbf{k}\,=\, 4\gamma^{3/2} \,\pi^{-1/2}\,\sum_{\nu=1}^N(-\imath \nabla_{x_\nu} \cdot \epsilon(k,p)) 
e^{- \imath\,\gamma^{1/2} k x_\nu}\,(2|k|)^{-1/2}\, 
\one[\,|k|\,\leq\, \kappa]+ \hc,\\
G^{i,\,\nu}_\mathbf{k}\, = \, H^{i,\,\nu}_\mathbf{k}
\,=\,2\gamma^{3/2} \,\pi^{-1/2}\,\epsilon_i(k,\,p)\, e^{- \imath\,\gamma^{1/2}\, k x_\nu}\,(2|k|)^{-1/2}\,  
\one[\,|k|\,\leq\, \kappa]
\end{gather*}  
for $i\,=\,1,\,2,\,3,\  \nu\,=\,1,\ldots,\, N,\ x_\nu\in \R^3$
and $\mathbf{k}=(k,p)\in \R^3\times\{\pm\}$. $\epsilon_i(k,\,\pm)\in\R^3 $ are polarization vectors.
\smallskip\\
$\blacktriangleright$ The Pauli-Fierz-Model: \\
 $\Hel\subset L^2(\R^{3N})$, $\Hael\,:=\, -\Delta\,+\,V$,\, 
 $\hph\,=\,L^2(\R^3)$ or $\hph\,=\,L^2(\R^3\times \{\pm\})$, and $\alpha(k)\,=\,|k|$.
 The form factor is $F_k\,=\,\gamma \sum_{\nu\,=\,1}^N \one[\,|k|\,\leq\, \kappa] \, k\cdot x_\nu$ and
 $G^j_k=H^j_k=0$
\bigskip
\section{The Representation $\pig$}\label{sec:Rep}
In order to describe the particle system at inverse temperature $\beta$ 
we introduce the algebraic setting. For $\f\, =\,\{ f\in \hph\,:\, \alpha^{-1/2}f \in \hph\}$
we define the algebra of observables by
\begin{equation*}\label{Algebra}
\Ag \,=\, \mathcal{B}(\Hel)\tensor \Wop(\f).
\end{equation*}
For elements $A\in \Ag$ we
define $\tilde{\tau}^0_t(A)\,:=\, e^{\imath \,t\, \Ho}\,A\, e^{-\imath\, t\, \Ho}$ and\\
$\tilde{\tau}^g_t(A)\,:=\, e^{\imath \,t\,\Hg}A e^{-\,\imath\, t\, \Hg}$.
We first discuss the model without interaction.
%
%
\subsection{The Representation $\pif$} 
The time-evolution for the Weyl operators is given by
\begin{equation*} \label{TimeEvField}
e^{\imath\, t \,\Haf}\,\Wop(f)\,e^{-\imath\, t\, \Haf}\,=\, \Wop(e^{\imath\, t\, \alpha }\,f).
\end{equation*}
For this time-evolution an equilibrium state $\omf$ is defined by
\begin{equation*}
\omf(\Wop(f))\,=\, \langle\, f \,|\, (1 \,+\, 2\,\varrho_\beta) \,f\,\rangle_{\hph},
\end{equation*}
where $\varrho_\beta(k)\,=\, \big(\exp(\beta \,\alpha(k)) \,-\,1\big)^{-1}$.
It describes an infinitely extended gas of bosons with momentum density
$\varrho_\beta$ at temperature $\beta$.
Since $\omf$ is a quasi-free state on the Weyl algebra, the definition
of $\omf$ extends to polynomials of creation and annihilation operators.
One has
\begin{gather*} \label{EvOnePointTwoPoint}
\omf(a(f))\,=\,\omf(a^*(f))\,=\,\omf(a(f)\,a(g))\,=\,\omf(a^*(f)\,a^*(g))\,=\,0,\\
\omf(a^*(f)\,a(g))\,=\, \langle\, g\,|\, \varrho_\beta \,f \,\rangle_{\hph}.
\end{gather*}
For polynomials of higher degree one can apply Wick's theorem for quasi-free states, i.e.,
\begin{equation} \label{QuasiFreeExp}
\omf(a^{\sigma_{2m}}(f_{2m})\cdots a^{\sigma_1}(f_1))
\,=\,\sum_{P\in \mathcal{Z}_2}\prod_{\stackrel{\{i,j\}\in P}{{i>j}}}
\omf\big(a^{\sigma_i}(f_i)\,a^{\sigma_j}(f_j)\big),
\end{equation}
where $a^{\sigma_{k}}\,=\,a^*$ or $a^{\sigma_{k}}\,=\,a$ for $k=1,\ldots,\,2m$.
$\mathcal{Z}_2$ are the pairings, that is
\begin{center}
 $P\in\mathcal{Z}_2,$ iff
$P\,=\,\{ Q_1,\ldots,Q_m\},\ \#Q_i=2$ and $\bigcup_{i=1}^m\,Q_i\,=\,\{1,\ldots,\,2m\}$.\\
\end{center}
The Araki-Woods isomorphism $\pif\,:\, \Wop(\f)\,\rightarrow\, \mathcal{B}(\Fb\tensor \Fb)$
is defined by
\begin{gather*} \label{ArakiWoodsIso}
\pif[\Wop(f)] \,:=\, \Wop_\beta(f)\,:=\, \exp( \imath\, \Phi_\beta(f)),\\
\Phi_\beta(f)\,:=\, \Phi( (1\,+\,\varrho_\beta)^{1/2}\, f)\tensor \one 
\,+\, \one \tensor \Phi( \varrho_\beta^{1/2}\, \ovl{f}).
\end{gather*}
The vector $\Omf\,:=\, \Omega\tensor \Omega$ fulfills
\begin{equation}\label{AWRep}
\omf(\Wop(f))\,=\, \langle\, \Omf\,|\, \pif[\Wop(f)]\, \Omf \,\rangle.
\end{equation}
\subsection{The representation $\piel$}
The particle system without interaction has the observables
$\mathcal{B}(\Hel)$, the states are defined by density operators $\rho$,
i.e., $\rho \in \mathcal{B}(\Hel),\ 0\leq \rho,\ \Tr\{\rho\}\,=\,1$. The expectation
of $A\in \mathcal{B}(\Hel)$ in $\rho$ at time $t$ is
\begin{equation*}
\Tr\{\, \rho\, e^{\imath\, t\, \Hael}\,A\, e^{-\imath\, t\, \Hael}\}.
\end{equation*}
Since $\rho$ is a  compact, self-adjoint operator, there is
an ONB $(\phi_n)_n$ of eigenvectors, with corresponding (positive) eigenvalues
$(p_n)_n$. Let 
\begin{equation}\label{Vektor-rep}
\sigma(x,\,y) \,=\,\sum_{n=1}^\infty p_n^{1/2}\, \phi_n(x)\, \ovl{\phi_n(y)}\in \Hel\tensor \Hel.
\end{equation}
For $\psi \in \Hel$ we define $\sigma\, \psi\,:=\, \int \sigma(x,\,y)\,\psi(y)\, d\mu(y)$. 
Obviously, $\sigma$ is an operator of Hilbert-Schmidt class.
Note, $\ovl{\sigma}\,\psi \,:=\, \ovl{\sigma\, \ovl{\psi}}$ has the integral kernel $\ovl{\sigma(x,\,y)}$.
It is a straightforward calculation to verify that
\begin{equation*}
\Tr\{ \rho\, e^{\imath\, t\, \Hael}\,A\, e^{-\imath\, t \,\Hael}\}
\,=\, \langle\, e^{-\imath\, t \,\Lel}\,\sigma \,|
\, (A\tensor 1)\, e^{-\imath\, t\, \Lel}\sigma\,\,\rangle_{\Hel\tensor \Hel},
\end{equation*}
where $\Lel\,=\, \Hael\tensor \one - \one \tensor \ovl{H}_{el}$. This suggests the definition of
the representation
\begin{equation*}
\piel\,:\, \mathcal{B}(\Hel)\rightarrow \mathcal{B}(\Hel\tensor \Hel),\quad A\mapsto A\tensor \one.
\end{equation*}
Now, we define the representation map for the joint system by
\begin{equation*}
\pi\,:\, \Ag\, \rightarrow\, \mathcal{B}(\Kg),\qquad \pi\,:=\, \pi_{el}\tensor \pi_f,
\end{equation*}
where $\Kg\,:=\,\Hel \tensor \Hel \tensor \Fb\tensor \Fb$.
Let $\Mg\,:=\, \pi[\Ag]''$ be the \emph{enveloping $W^*$-algebra}, here $\pi[\Ag]'$ denotes
the commutant of $\pi[\Ag]$, and $\pi[\Ag]''$ the bicommutant. 
We set $\Core:= U_1\otimes \overline{U_1}\otimes \mathcal{C}$,
where $\mathcal{C}$ is a subspace of vectors in $\mathcal{F}_{b}^{fin}\otimes\mathcal{F}_{b}^{fin}$,
with compact support, and $U_1:=\cup_{n=1}^\infty \operatorname{ran}\,\mathbbm{1}[\Hel \le n]$.
On $\Core$ the operator $\Lo$, given by
\begin{eqnarray*}
\Lo &\,:=\,& \Lel\tensor \one \,+\, \one \tensor \Lf,\quad \textrm{on } \Kg,\\
\Lf &\,:=\,& d\Gamma(\alpha)\tensor \one \,-\, \one \tensor d\Gamma(\alpha),\quad \textrm{on } \Fb\tensor \Fb,
\end{eqnarray*}
is essentially self-adjoint and we can define 
\begin{equation*}
\tauo_t(X)\,:=\, e^{\imath \,t\, \Lo}\,X\, e^{-\imath \,t\, \Lo}\in \Mg,\quad X\in \Mg,\quad t\in \R,
\end{equation*}
It is not hard to see, that 
\begin{equation*}
\pi[\tilde{\tau}^0_t(A)]\,=\,\tauo_t(\pi[A]),\quad A\in \Ag,\,\quad t\in \R
\end{equation*}
On $\Kg$ a we introduce a conjugation  by 
\begin{equation*}
\Jg\,( \phi_1 \tensor \phi_2 \tensor \psi_1 \tensor \psi_2)
\,=\,\ovl{\phi_2} \tensor \ovl{\phi_1} \tensor \ovl{\psi_2} \tensor \ovl{\psi_1 }.
\end{equation*}
It is easily seen, that $\Jg\, \Lo\,=\, -\Lo\, \Jg$. 
In this context one has $\Mg'\,=\,\Jg\,\Mg\,\Jg$,
see for example \cite{BratteliRobinson1987}. In the case, where
$\Hael$ fulfills Hypothesis \ref{Hyp1}, we define the vector representative $\Omel
\in \Hel\tensor \Hel$ of the Gibbs state $\omel$ as in \eqref{Vektor-rep}
for $\rho= e^{-\beta \Hael}\,\Zel^{-1}$.
%
%
\begin{Satz}\label{FreeSystem}
Assume Hypothesis \ref{Hyp1} is fulfilled. Then,
$\Omo\,:=\, \Omel\tensor \Omf$ is a \emph{cyclic} and \emph{separating} vector for $\Mg$. 
$e^{-\beta/2 \Lo}$ is a \emph{modular operator} and $\Jg$ is the \emph{modular conjugation}
for $\Omo$, that is
\begin{equation}
X\Omo\in \dom(e^{-\beta/2\, \Lo}),\quad \Jg \,X \,\Omega\,=\, e^{-\beta/2 \,\Lo}\,X^*\,\Omo
\end{equation}
for all $X\in\Mg$ and $\Lo\, \Omo\,=\,0$. Moreover,  
\begin{equation*}
\omo(X)\,:=\, \langle\, \Omo\,|\, X\,\Omo\,\rangle_{\Kg},\qquad X\in \Mg
\end{equation*}
is a $(\tauo,\,\beta)$-KMS-state for $\Mg$, i.e., for all $X,\,Y\in \Mg$
exists $F_\beta(X,\,Y,\cdot)$, analytic in the strip $S_\beta\,=\,\{ z\in \C\,:\, 0\,<\,\Im z\, <\beta\}$,
continuous on the closure and taking the boundary conditions
\begin{eqnarray*}
F_\beta(X,\,Y,\,t)&\,=\,&\omo(X\,\tauo_t(Y))\\
F_\beta(X,\,Y,\,t\,+\,\imath\,\beta)&\,=\,&\omo(\tauo_t(Y)\,X)
\end{eqnarray*}
\end{Satz}
%
%
For a proof see \cite{JaksicPillet1996b}.
%
\bigskip
\section{The Liouvillean $\Lg$}
%
In this and the next section we will introduce the Standard Liouvillean 
$\Lg$ for a dynamics $\taug$ on $\Mg$, describing the interaction between 
particles and bosons at inverse temperature $\beta$. The label $Q$ denotes the 
interaction part of the Liouvillean, it can be deduced from the interaction part 
$W$ of the corresponding Hamiltonian by means of formal arguments, which we will
not give here. In a first step we prove self-adjointness of $\Lg$ and 
of other Liouvilleans. A main difficulty stems from the fact, that $\Lg$
and the other Liouvilleans, mentioned before, are not bounded from below. The
proof of self-adjointness is given in Theorem \ref{Thm1}, it uses Nelson's commutator
theorem and auxiliary operators which are constructed in Lemma \ref{Lem1.1x}.
The proof, that $\taug_t(X)\in \Mg$ for $X\in \Mg$, is given in Lemma \ref{LemInv}.
Assuming $\Omo\in \dom(e^{-\beta/2(\Lo+\Qg)})$ we can ensure existence
of a $(\taug,\beta)$-KMS state $\omg(X)=\langle \Omg \,|X\,\Omg \rangle\cdot \|\Omg\|^{-2}$ 
on $\Mg$, where $\Omg=e^{-\beta/2(\Lo+\Qg)}\Omo$. Moreover, we can show
that $e^{-\beta\Lg}$ is the modular operator for $\Omg$ and conjugation $\Jg$.
This is done in Theorem \ref{Thm3}.\\
\indent Our proof of \ref{Thm3} is inspired by the proof given in \cite{DerezinskiJaksicPillet2003}.
The main difference is that we do not assume, that $\Qg$ is self-adjoint and
that $\Omo\in \dom(e^{-\beta \Qg})$. For this reason we need to introduce
an additional approximation $\QN$ of $\Qg$, which is self-adjoint and
affiliated with $\Mg$, see Lemma \ref{Lem5d}.

The interaction on the level
of Liouvilleans between particles and bosons is given by $\Qg$ , where
\begin{equation*}
\Qg \,:=\, \Phi_\beta(\vec G )\,\Phi_\beta(\vec H )\,+\,\hc\,+\,\Phi_\beta(F),
\quad \Phi_\beta(\vec G )\,\Phi_\beta(\vec H )\,:=\, \sum_{j=1}^r \Phi_\beta( G^j )\,\Phi_\beta(H^j ).
\end{equation*}
For each family $K=\{K_k\}_{k}$ of closed operators on $\Hel$ with
$\|K\|_{w,1/2}<\infty$ we set
\begin{equation*}
\Phi_\beta(K)\,:= \,
\big(a^*( (1\,+\,\varrho_\beta)^{1/2}\, K)\tensor \one 
\,+\, \one \tensor a^*( \varrho_\beta^{1/2}\, K^*)\big) \,+\, \hc.
\end{equation*}
Here, $K_k$ acts as $K_k\otimes \one$ on $\Hel\otimes \Hel$.
A Liouvillean, that describes the dynamics of the joint system of particles
and bosons is the so-called \emph{Standard Liouvillean}
\begin{equation}\label{Def:Lg}
\Lg\,\phi \,:=\, (\Lo \,+\, \Qg \,-\, \Qg^ \Jg)\,\phi,\qquad \phi\in \Core,
\end{equation}
which is distinguished by $\Jg\, \Lg=-\Lg\, \Jg$. For an operator $A$, acting on $\Kg$,
the symbol $A^\Jg$ is an abbreviation for $\Jg\,A\,\Jg$.
An important observation is, that $[ \Qg\,,\, \Qg^\Jg]\,=\,0$ on $\Core$.
Next, we define four auxiliary operators on $\Core$
\begin{align}\label{eq3.1x}
\Laux^{(1)}&\,:=\, (\Haelp\tensor \one \,+\, \one \tensor \ovl{H}_{el,+}\big)\tensor \one
+\one\tensor\Lfaux+\one\\ \nonumber
\Laux^{(2)}&\,:=\, \Haelp^Q+(\Haelp^Q)^{ \Jg} +c_1\one\tensor\Lfaux+c_2\\ \nonumber
\Laux^{(3)}&\,:=\, \Haelp^Q+ (\Haelp)^{\Jg}+c_1\one\tensor\Lfaux+c_2\\ \nonumber
\Laux^{(4)}&\,:=\, \Haelp\tensor \one + (\Haelp^Q)^{ \Jg}+c_1\one\tensor\Lfaux+c_2,
\end{align}
where $\Lfaux$ is an operator on $\Fb\tensor \Fb$ and $\Haelp^Q$
acts on $\Kg$. Furthermore,
\begin{align*}
\Lfaux \,&=\, d\Gamma(1+\alpha)\tensor \one + \one \tensor d\Gamma(1+\alpha)+\one,\\
\Lelaux\,&=\,\Haelp\tensor \one \,+\, \one \tensor \ovl{H}_{el,+}
\quad \Haelp^Q\,:=\, \Haelp\tensor \one+\Qg.
\end{align*}
Obviously, $\Laux^{(i)},\ i\,=\,1,\,2,\,3,\,4$ are symmetric operators on 
$\Core$.
%
%
\begin{Lemma}\label{Lem1.1x}
For sufficiently large values of $c_1,\,c_2\,\geq \,0$ 
we have that $\Laux^{(i)},\ i\,=\,1,\,2,\,3,\,4$ are
essentially self-adjoint and positive.
Moreover, there is a constant $c_3\,>\,0$ such that
\begin{equation}\label{eq3.2}
c_3^{-1}\,\|\Laux^{(1)}\,\phi\|\,\leq \,\|\Laux^{(i)}\,\phi\|
\leq c_3\,\|\Laux^{(1)}\,\phi\|,\qquad \phi \in\dom(\Laux^{(1)}).
\end{equation}
\end{Lemma}
%
%
\begin{proof}
Let $a,\,a'\in \{l,\,r\}$ and $K_i,\ i\,=\,1,\,2$ be families of bounded operators 
with $\|K_i\|_{w}<\infty$.
Let $\Phi_l(K_i)\,=\,\Phi(K_i)\tensor \one $ and $\Phi_r(K_i)\,:=\,\one\tensor \Phi(K_i).$
We have for $\phi \in \mathcal{D}$
\begin{align} 
\|\,\Phi_a(\eta K_1)\,\Phi_{a'}(\eta'K_2)\,\phi\,\|
\,&\leq\, \const \|\,\Lfaux\,\phi\,\|\\ \nonumber
\|\,\Phi_a(\eta F)\,\phi\,\| \,&\leq\, \const \| (\Lelaux)^{1/2}(\Lfaux)^{1/2}\,\phi\,\|,
\end{align}
where $\eta,\,\eta'\in\{ (1+\varrho_\beta)^{1/2},\, \varrho^{1/2}_\beta\}$.
Note, that the estimates hold true, if $\Phi_a(\eta K_i)$ or $\Phi_a(\eta F)$ are replaced by 
${\Phi_a(\eta K_i)}^{\Jg}$ or ${\Phi_a(\eta F)}^{\Jg}$.
Thus, we obtain for sufficiently large $c_1\,\gg\, 1 $, depending on the form-factors, that
\begin{equation}\label{eq3.6}
\|\Qg\,\phi\|\,+\,\|{\Qg}^{\Jg}\,\phi\|\,\leq \,1/2\, 
\big\| \big(\Lelaux+c_1\,\Lfaux\big)\,\phi\big\|.
\end{equation}
By the Kato-Rellich-Theorem ( \cite{ReedSimonII1980}, Thm. X.12) we deduce that
$\Laux^{(i)}$ is self-adjoint  on $\dom(\Lelaux+c_1\,\Lfaux)$,
bounded from below and that $\Lelaux+c_1\,\Lfaux$ is $\Laux^{(i)}$-bounded for 
every $c_2\,\geq\, 0$ and $i\,=\,2,\,3,\,4$. In particular, $\Core$ is
a core of $\Laux^{(i)}$. The proof follows now from $\|\Laux^{(i)}\,\phi\|\,\le\, \|(\Lelaux+c_1\,\Lfaux)\,\phi\|\,\le 
\,c_1\, \|\Laux^{(1)}\,\phi\|$ for $\phi\in \Core$.
\end{proof}
%
%
\begin{Satz}\label{Thm1}
The operators
\begin{equation}\label{eq3.7}
\Lo,\quad \Lg\,=\, \Lo+\Qg-{\Qg}^{\Jg},
\quad\Lo+\Qg,\quad  \Lo-{\Qg}^{\Jg},
\end{equation}
defined on $\mathcal{D}$, are essentially self-adjoint. Every
core of $\Laux^{(1)}$ is a core of the operators in line \eqref{eq3.7}. 
\end{Satz}
%
%
\begin{proof}
We restrict ourselves to the case of $\Lg$.
We check the assumptions of Nelson's commutator theorem (\cite{ReedSimonII1980}, Thm. X.37). 
By Lemma \ref{Lem1.1} it suffices to show 
$\|\Lg\phi\|\,\leq \,\const \|\Laux^{(1)}\phi\|$ and 
$|\langle\, \Lg\phi\,|\,\Laux^{(2)}\phi \rangle 
-\langle\, \Laux^{(2)}\phi|\,\Lg\phi \rangle |
\,\leq \,\const  \|(\Laux^{(1)})^{1/2}\phi\|^2$
for $\phi\in \mathcal{D}$. The first inequality follows from Equation \eqref{eq3.6}. 
To verify the second inequality we observe
\begin{align}\label{eq3.9} 
\big|&\big\langle\, \Lg\,\phi\,\big| \,\Laux^{(2)}\,\phi \,\big\rangle 
-\big \langle\, \Laux^{(2)}\,\phi\,\big| \,\Lg\,\phi \,\big\rangle \big|\\ \nonumber
&\leq
c_1 \Big|\big\langle\,  \Qg \phi\,\big|\,\Lfaux \phi \big\rangle
-\big \langle\, \Lfaux \phi\, \big|\, \Qg \phi\big\rangle\Big|\\ \nonumber
&\phantom{\leq}+ c_1 \Big|\big\langle\,  {\Qg}^{\Jg} \phi\,\big|\, \Lfaux\phi\,\big\rangle
- \big\langle\, \Lfaux \phi\,\big|\,{\Qg}^{\Jg} \phi\,\big\rangle\Big|\\  \nonumber
&\phantom{\leq} + \Big|\big\langle\, \Lf \phi \big|\,\Qg \phi\,\big\rangle
- \langle\,  \Qg\phi\,\big|\, \Lf\phi\,\big\rangle\Big|
+  \Big|\big\langle\, \Lf \phi\,\big|\,{\Qg}^{\Jg} \phi\,\big\rangle
-\big\langle\,  {\Qg}^{\Jg} \phi\,\big|\, \Lf \phi\,\big\rangle \Big|,
\end{align}
where we used, that  
$
\big[\Haelp^Q ,\, (\Haelp^Q)^{\Jg}\big]\,=\,0.
$
Let $K_i\in \{ G_j,\,H_j\}$ and $\eta,\,\eta' \in \{\varrho^{1/2},\, (1\,+\,\varrho)^{1/2}\}$.
We remark, that 
\begin{align} \label{eq3.10}
[\Phi_a(\eta\, K_1)\,\Phi_{a'}(\eta'\, K_2)\,,\,\Lfaux]
&\,=\,\imath\, \Phi_a(\imath \,(1\,+\,\alpha)\, \eta\, K_1)\,\Phi_{a'}v(\eta'\,K_2)\\ \nonumber
&\,+\,\imath \,\Phi_a(\eta\, K_1)\,\Phi_{a'}(\imath \,(1\,+\,\alpha)\, \eta'\,K_2)\\ \nonumber
[\Phi_a(\eta \,F)\,,\ \Lfaux]&\,=\, \imath \,\Phi_a(\imath \,(1\,+\,\alpha) \,\eta\, F).
\end{align}
Hence, for $\phi\in \dom(\Laux^{(2)})$, we have by means of \eqref{RelBoundGen} that
\begin{align} \label{eq3.11}
\big|\big \langle\, \phi\,|\, [\Phi_a(\eta K_1)\,\Phi_{a'}(\eta' K_2),\ \Lfaux]\phi \,\big \rangle \big|
\,&\leq\, \const \|\Lfaux^{1/2}\phi \|^2\\ \nonumber
\big|\big \langle \phi\,|\, [\Phi_a(\eta F),\ \Lfaux]\phi\, \big \rangle \big|
\,&\leq\, \const \|\Lfaux^{1/2}\phi \|\,\| (\Lelaux)^{1/2}\phi\|.
\end{align}
Thus, \eqref{eq3.11} is bounded by a constant times $ \|(\Laux^{(1)})^{1/2}\phi\|^2$.
The essential self-adjointness of $\Lg$ follows now from estimates analog
to \eqref{eq3.10} and \eqref{eq3.11}, where $\Lfaux$ is replaced by $\Lf$ in
\eqref{eq3.10} and in the left side of $\eqref{eq3.11}$.
For $\Lo+\Qg$ and $\Lo-{\Qg}^{\Jg}$ one has to consider
the commutator with $\Laux^{(3)}$ and $\Laux^{(4)}$, respectively.
\end{proof}
%
%
\begin{Bem}\label{RemarkSelfadjoint} In the same way one can show, that
$\Hg$ is essentially self-adjoint on any core of $H_1:=\Hael\,+\, d\Gamma(1\,+\,\alpha)$,
even if $\Hg$ is not bounded from below. 
\end{Bem}
%
%
%
\section{Regularized Interaction and Standard Form of $\Mg$}\label{RegVer}
In this subsection a regularized interaction $\QN$ is introduced:
\begin{equation}\label{DefQN}
\QN\,:=\,  \Big\{\Phi_\beta(\vec G_{N})\,\Phi_\beta(\vec H_{N})\,+\, \hc\Big\}
\,+\,\Phi_\beta(F_{N}).
\end{equation}
The regularized form factors $\vec G_{N},\,\vec H_{N},\,F_{N}$ are obtained by multiplying the finite rank projection
$P_N:=\mathbf{1}[\Hael\leq N]$ from the left and the right. Moreover, an additional
ultraviolet cut-off $\one[\alpha\leq N]$, considered as a spectral projection, is added. 
The regularized form factors are
\begin{gather*}\nonumber
\vec G_{N}(k)\,:= \,\one[ \alpha \,\leq \,N]\, P_N \,\vec G(k)\,P_N,\qquad 
\vec H_{N}(k)\,:=\, \one[ \alpha \,\leq\, N]\,P_N \,\vec H(k)\,P_N,\\ F_{N}(k)\,:= \,\one[  \alpha\, \leq\, N]\,P_N\,F(k)\,P_N.
\end{gather*}
%
%
\begin{Lemma}\label{Lem5d}
i) $\QN$ is essentially self-adjoint on $\mathcal{D} \subset\dom(\QN)$.
$\QN$ is affiliated with $\Mg$, i.e,. $\QN$ is closed and 
\begin{equation*}
X'\,\QN \subset \QN \,X',\quad \forall\, X'\in \Mg'.
\end{equation*}
ii) $\Lo+\QN$, $\Lo-\Jg\QN\Jg $ and
$\Lo+\QN-\Jg\QN\Jg $ converges in the strong resolvent sense to $\Lo+\Qg$,
$\Lo-\Jg\Qg\Jg $ and $\Lo+\Qg-\Jg\Qg\Jg $, respectively.\\
\end{Lemma}
%
%
\begin{proof}
Let $\QN$ be defined on $\mathcal{D}$. 
With the same arguments as in the proof of Theorem \ref{Thm1} we obtain
\begin{equation*}
\|\QN\phi\|\,\leq\, C \|\Lfaux\phi\|,
\quad \big|\big\langle\, \QN\phi\,\big|\,\Lfaux\phi \big\rangle 
\,-\,\big\langle\, \Lfaux\phi\,\big|\,\QN\phi \big\rangle \big|\,
\leq \,C  \big\|(\Lfaux)^{1/2}\phi\big\|^2,
\end{equation*}
for $\phi \in \mathcal{D}$ and some constant $C>0$, where we have used that $\|F_N\|_{w}\,<\,\infty$.
 Thus, from Theorem \ref{Thm1} and Nelson's commutator theorem
we obtain that $\mathcal{D}$ is a common core for
$\QN,\ \Lo+\QN,\ \Lo\,-\,\QN^{\Jg},\ \Lo\,+\,\QN\,-\QN^{\Jg}$
and for the operators in line \eqref{eq3.7}.
A straightforward calculation yields 
\begin{equation*}
\lim_{N\rightarrow\infty}\,\QN\phi
\,=\,\Qg \phi,\quad
\lim_{N\rightarrow\infty}\,\Jg\QN\Jg\phi
\,=\,\Jg\Qg \Jg\phi
\qquad \forall\, \phi\in \mathcal{D}.
\end{equation*}
Thus statement ii) follows.\\
\indent Let $N_f\,:=\, d\Gamma(1)\tensor \one + \one \tensor d\Gamma(\one)$ be the number-operator.
Since $\dom(N_f)\supset \mathcal{D}$  and ${\Wop_\beta(f)}^{\Jg}\,:\,\dom(N_f)\,\rightarrow\, \dom(N_f)$,
see \cite{BratteliRobinson1987}, we obtain
\begin{equation}
\QN (A\tensor\one \tensor \Wop_\beta(f))^{\Jg} \phi\,=\, (A\tensor\one \tensor \Wop_\beta(f))^{\Jg}\QN \phi
\end{equation}
for $A\in\mathcal{B}(\Hel),\ f\in \f$ and $\phi \in \mathcal{D}$. 
By closedness of $\QN$ and density arguments the equality holds for $\phi \in \dom(\QN)$
and $X\in\Mg$ instead of $A\tensor\one \tensor \Wop_\beta(f)$.
Thus $\QN$ is affiliated with $\Mg$ and therefore $e^{\imath \,t \QN}\in \Mg$ for
$t\in \R$.\\
\end{proof}
%
%
\begin{Lemma}\label{LemInv}
We have for $X\in \Mg$ and $t\in\R$
\begin{align}
\taug_t(X)= e^{\imath t(\Lo+\Qg)}\,X\,e^{\imath t(\Lo+\Qg)},\quad
\tauo_t(X)= e^{\imath t(\Lo-\Qg^\Jg)}\,X\,e^{\imath t(\Lo-\Qg^\Jg)}
\end{align}
\indent Moreover, $\taug_t(X)\in \Mg$ for all $X\in \Mg$ and $t\in \R$, such
as 
$$E_{\Qg}(t)\,:= \,e^{\imath \,t \,(\Lo\,+\,\Qg)}\,e^{-\imath\, t\,\Lo}=
e^{\imath \,t \,\Lg}\,e^{-\imath\, t\,(\Lo-\Qg^{\Jg})}
\in \Mg.$$
\end{Lemma}
\begin{proof}
First, we prove the statement for $\QN$, since $\QN$ is affiliated with $\Mg$ and therefore $e^{\imath t\QN}\in \Mg$.
We set
\begin{gather}
\hat{\tau}_t^N(X)=e^{\imath\, t\,(\Lo\,+\,\QN)}\,X \,e^{-\imath \,t\,(\Lo\,+\,\QN)},\quad
\hat{\tau}_t(X)=e^{\imath\, t\,(\Lo\,+\,\Qg)}\,X \,e^{-\imath \,t\,(\Lo\,+\,\Qg)}
\end{gather}
On account of Lemma \ref{Lem5d} and Theorem \ref{Thm1} we can apply
the Trotter product formula to obtain
\begin{align*}
\hat{\tau}_t^N(X)& 
= \wlim_{n\,\to \,\infty}
\big(e^{\imath \,\frac{t}{n}\Lo}\,e^{\imath \,\frac{t}{n}\QN} \big)^n  
X \big(e^{-\imath\, \frac{t}{n}\QN}\,e^{-\imath \,\frac{t}{n}\Lo} \big)^n\\ \nonumber
&= \wlim_{n\to \infty} 
\tauo_{\frac{t}{n}}\big(
e^{\imath \,\frac{t}{n}\,\QN}\cdots 
\tauo_{\frac{t}{n}}(e^{\imath\, \frac{t}{n}\,\QN}\,X\, e^{-\imath \,\frac{t}{n}\QN})
\cdots e^{-\imath \,\frac{t}{n}\QN}
\big).
\end{align*}
Since $e^{\imath \,\frac{t}{n}\QN}, X\in \Mg$ and since $\tauo$ leaves $\Mg$
invariant,  $\hat{\tau}^N_t(X)$ is the weak limit of elements of $\Mg$,
and hence $\hat{\tau}^N_t(X)\in \Mg$. 
Moreover,
\begin{equation*}
 \hat{\tau}_t(X)\,=\,\wlim_{N\,\rightarrow\, \infty} \hat{\tau}^N_t(X)\in \Mg.
\end{equation*}
For $E_N(t)\,:=\, e^{\imath \,t\,(\Lo\,+\,\QN)}e^{-\imath \,t \,\Lo} \in \mathcal{B}(\Kg)$
we obtain
\begin{align*}
e^{\imath t (\Lo\,+\,\QN)}e^{-\imath \,t\,\Lo}
&= \slim_{n\rightarrow \infty}\big(e^{\imath\, \frac{t}{n}\, \Lo}
    e^{\imath\, \frac{t}{n}\, \QN}\big)^n\,e^{-\imath \,t\, \Lo}\\ \nonumber
&=\slim_{n\rightarrow \infty}\tauo_{\frac{t}{n}}(e^{\imath \,\frac{t}{n}\, \QN})
    \tauo_{\frac{2t}{n}}(e^{\imath\, \frac{t}{n}\, \QN}) \cdots
    \tauo_{\frac{nt}{n}}(e^{\imath \,\frac{t}{n}\,\QN})\in \Mg.
\end{align*}
By virtue of Lemma \ref{Lem5d} we get
$E_{\Qg}(t)\,:=\, e^{\imath \,t\,(\Lo\,+\,\Qg)}\,e^{-\imath \,t \,\Lo}\,
=\,\wlim_{N\,\rightarrow\,\infty}\, E_N(t)\in \Mg$.
Since $\Jg$ leaves $\Core$ invariant and thanks to Lemma \ref{Lem5d},
we deduce, that $\Core$ is a core of $\Jg\QN\Jg$. Moreover, we have
$ e^{-\imath t \QN^\Jg}=\Jg e^{\imath t \QN}\Jg\in \Mg'$. Since 
we have shown, that $\hat{\tau}^N$ leaves $\Mg$ invariant, we get
\begin{align*}
\taun_t(X)&= \wlim_{n\to \infty} 
(e^{\imath \frac{t}{n}(\Lo+\QN)}
e^{\imath \frac{t}{n}(-\QN^\Jg)})^n
X\,(e^{-\imath \frac{t}{n}(-\QN^\Jg)}\,
e^{-\imath \frac{t}{n}(\Lo+\QN)})^n \\
&=\wlim_{n\to \infty}\hat{\tau}^N_{\frac{t}{n}}\big(
e^{-\imath \,\frac{t}{n}\,\QN^\Jg}\cdots 
\hat{\tau}^N_{\frac{t}{n}}(e^{-\imath\, \frac{t}{n}\,\QN^\Jg}
\,X\, e^{\imath \,\frac{t}{n}\QN^\Jg})
\cdots e^{\imath \,\frac{t}{n}\QN^\Jg}\big)\\
&= \hat{\tau}^N_t(X).
\end{align*}
Thanks to Lemma \ref{Lem5d} we also have
\begin{equation}
\taug_t(X)=\wlim_{n\to\infty}\taun_t(X)=\wlim_{\N\to\infty}\hat{\tau}^N_t(X)=\hat{\tau}_t(X).
\end{equation}
The proof of $\tauo_t(X)= e^{\imath t(\Lo-\Qg^\Jg)}\,X\,e^{\imath t(\Lo-\Qg^\Jg)}$ follows
analogously.
Using the
Trotter product formula we obtain
\begin{eqnarray*}
e^{\imath t (\Lo+\QN)}\,e^{-\imath t \Lo}
&\,=\,& \slim_{n\,\to\, \infty}\big(e^{\imath \frac{t}{n}\, \Lo}
    e^{\imath \frac{t}{n}\QN}\big)^n\,e^{-\imath t \Lo}\\ \nonumber
&=&\slim_{n\,\to\, \infty} \tauo_{\frac{t}{n}}(e^{\imath \frac{t}{n} \QN})
    \tauo_{\frac{2t}{n}}(e^{\imath \frac{t}{n} \QN}) \cdots
    \tauo_{\frac{nt}{n}}(e^{\imath \frac{t}{n} \QN})\\ \nonumber
&=& \slim_{n\,\to \,\infty}\big(e^{\imath \frac{t}{n}(\Lo-{\QN}^{\Jg})}
    e^{\imath \frac{t}{n} \QN}\big)^n\,e^{-\imath t (\Lo-{\QN}^{\Jg})}\\ \nonumber
&=& e^{\imath t (\Lo +\QN-\Jg\QN\Jg)}\,e^{-\imath t (\Lo-{\QN}^{\Jg})}.
\end{eqnarray*}
By strong resolvent convergence we may deduce  $E(t)=e^{\imath t \Lg}\,e^{-\imath t (\Lo-{\Qg}^{\Jg})}$ .
\end{proof}
%
%
%
%
%
Let $\mathcal{C}$ be the natural positive cone associated with $\Jg$ and $\Omo$
and let $\Mana$ be the $\taug$-analytic elements of $\Mg$, (see \cite{BratteliRobinson1987}).
\begin{Satz} \label{Thm3}
Assume Hypothesis \ref{Hyp1} and $\Omo\in \dom(e^{-\beta/2\,(\Lo\,+\,\Qg)})$. Let $\Omg\,:=\, e^{-\beta/2\,(\Lo\,+\,\Qg)}\,\Omo$. 
Then
\begin{gather}
\Jg\, \Omg\,=\,\Omg,\qquad
\Omg\,=\, e^{\beta/2\,(\Lo\,-\,{\Qg}^{\Jg})}\,\Omo,\\ \nonumber
\Lg\,\Omg\,=\,0,\qquad
\Jg\, X^*\,\Omg \,= \,e^{-\beta/2\,\Lg}X\,\Omg,\quad \forall\,X\in\Mg
\end{gather}
Furthermore, $\Omg$ is separating and cyclic for $\Mg$, and $\Omg \in\mathcal{C}$. 
The state $\omg$ is defined by
$$\omg(X):= \|\Omg\|^{-2}\,\langle \Omg\,|X\,\Omg\rangle,\ X\in \Mg$$
is a $(\taug,\, \beta)$-KMS state on $\Mg$.
\end{Satz}
%
%
\begin{proof}
First, we define $\Omega(z)\,=\, e^{-z\,(\Lo+\Qg)}\,\Omo$ for $ z\in \C$ with $0\le\Re z\le \beta/2$.
Since $\Omo \in \dom(e^{-\beta/2\,(\Lo+\Qg)})$, $\Omega(z)$ is analytic on $\mathcal{S}_{\beta/2}:=
\{ z\in \C\,:\, 0\,<\, \Re(z)\,<\,\alpha\}$
and continuous on the closure of $\mathcal{S}_{\beta/2}$, see Lemma \ref{Lem2App} below. \\
\newline
$\blacktriangleright$ Proof of $\Jg\,\Omega(\beta/2)\,=\,\Omega(\beta/2)$:\\
We pick $\phi \in \bigcup_{n\in\N} \ran \one[|\Lo|\,\leq \,n]$. Let  
$
f(z)\,:=\,\langle\, \phi\,|\, \Jg \,\Omega(\overline{z})\,\rangle$ and $
g(z)\,:=\,\langle\, e^{-(\beta/2\,-\,\overline{z})\,\Lo}\,\phi\,|\, e^{-z\,(\Lo\,+\,\Qg)}\,\Omo\,\rangle$. 
Both $f$ and $g$ are analytic on $\mathcal{S}_{\beta/2}$ and
continuous on its closure. Thanks to Lemma \ref{LemInv} we have $E_{\Qg}(t)\in \Mg$, and hence
\begin{equation*}
f(\imath t)\,=\, \langle\, \phi\,|\,   \Jg\, E_{\Qg}(t)\,\Omo\,\rangle
\,=\,\langle\, \phi\,|\, e^{-\beta/2 \,\Lo}\, E_{\Qg}(t)^*\,\Omo\,\rangle
\,=\, g(\imath \,t),\ t\in \R.
\end{equation*}
By Lemma \ref{Lem1App}, $f$ and $g$ are equal, in particular in $z\,=\,\beta/2$. 
Note that  $\phi$ is any element of a dense subspace.\\
\newline
$\blacktriangleright$ Proof of $\Omo \in \dom(e^{\beta/2\,(\Lo- {\Qg}^{\Jg})})$
and $\Omega(\beta/2)\,=\, e^{\beta/2\,(\Lo-{\Qg}^{\Jg})}\,\Omo$:\\
Let $\phi \in \bigcup_{n\in\N} \ran \one[|\Lo\,-{\Qg}^{\Jg}|\,\leq\, n]$.
We set
$g(z)\, :=\, \langle\, e^{\ovl{z}(\Lo\,-{\Qg}^{\Jg})}\,\phi\,|\, e^{-z\,\Lo}\,\Omo\,\rangle$.
Since ${E_{\Qg}(t)}^{\Jg}\,=\, e^{\imath \,t\,(\Lo-{\Qg}^{\Jg})}\,e^{-\imath \,t\,\Lo}$, 
$g$ coincides for  $z\,=\,\imath \,t$ with 
$f(z)\,:=\,\langle\, \phi\,| \,\Jg \,\Omega(\ovl{z})\,\rangle $. Hence
they are equal in $z\,=\,\beta/2$. The rest follows since
$e^{\beta/2\,(\Lo- {\Qg}^{\Jg})}$ is self-adjoint.\\
\newline
$\blacktriangleright$ Proof of $\Lg\,\Omega(\beta/2)\,=\,0$:\\
Choose $\phi \in \bigcup_{n\in\N} \ran \one[|\Lg|\,\leq \,n]$.
We define $g(z)\,:=\, \langle\, e^{-\ovl{z}\Lg}\phi \,|\,e^{z\,(\Lo-{\Qg}^{\Jg})}\,\Omo\,\rangle$
and $f(z)\,:= \, \langle\, \phi\,|\, \Omega(z)\,\rangle$ for $z$ in the closure
of $ \mathcal{S}_{\beta/2}$. Again both functions are equal
on the line $z\,=\,\imath \,t,\ t\in \R$. Hence $f$ and $g$ are identical,
and therefore $\Omega(\beta/2)\in \dom(e^{-\beta/2\, \Lg})$ and
$e^{-\beta/2 \Lg}\,\Omega(\beta/2)\,=\,\Omega(\beta/2)$.
We conclude that $\Lg\,\Omega(\beta/2)\,=\,0$.\\
\newline
$\blacktriangleright$ Proof of $\Jg\, X^*\,\Omega(\beta/2)\,
=\, e^{-\beta/2\Lg}\,X\,\Omega(\beta/2),\ \forall\,X\in \Mg$:\\
Fore $A\in \Mana$ we have, that 
\begin{align*} \nonumber
\Jg\, A^*\,\Omega(-\imath t)
\,&=\,\Jg \,A^*\,E_{\Qg}(t)\,\Omo
\,=\,e^{-\beta/2 \Lo}\, E_{\Qg}(t)^*A\Omo\\
&=\,e^{-(\beta/2\,-\,\imath \,t)\, \Lo}\,e^{-\imath t(\Lo+\Qg)} A\,\Omo\\
&=e^{-(\beta/2-\imath t) \Lo}\,\taug_{-t}(A)\,e^{-\imath t(\Lo+ \Qg)} \,\Omo.
\end{align*}
Let $\phi \in \bigcup_{n\in\N} \ran \one[|\Lo|\leq n]$.
We define
$
f(z)\,=\, \langle\, \phi\,|\,\Jg\, A^*\,\Omega(\ovl{z})\,\rangle$  and 
$g(z)\,=\, \langle\, e^{-(\beta/2-\ovl{z}) \Lo}\,\phi\,|\,
\taug_{\imath z}(A)\,\Omega(z)\,\rangle$.
Since $f$ and $g$ are analytic and equal for $z\,=\, \imath t$,
we have $\Jg A^*\Omega(\beta/2)\,=\,\taug_{\imath \, \beta/2}(A)\,\Omega(\beta/2)$.
To finish the proof we pick
$\phi \in \bigcup_{n\in\N} \ran \one[|\Lg|\,\leq \,n]$,
and set  
$f(z)\, :=\, \langle\, \phi\,|\,\taug_{\imath z}( A)\,\Omega(\beta/2)\,\rangle$
and $ g(z)\,:=\, \langle\, e^{-\overline{z}\Lg}\phi\,|\,A\Omega(\beta/2)\,\rangle$. 
For $z\,=\,\imath \,t$ we see 
\begin{equation*}
g(\imath t)\,=\,\langle\, \phi\,|
\, e^{-\imath t\Lg}\,A\, e^{\imath t\Lg}\Omega(\beta/2)\,\rangle
\,=\,\langle\, \phi\,|\, \taug_{-t}(A)\,\Omega(\beta/2)\,\rangle
\,=\,f(\imath \,t).
\end{equation*}
Hence $A\,\Omega(\beta/2)\in \dom(e^{-\beta/2\,\Lg})$ and
$\Jg A^*\Omega(\beta/2)\,=\, e^{-\beta/2\Lg}A\Omega(\beta/2)$.\\
Since $\Mana$ is dense in the strong topology, the equality holds for
all $X\in \Mg$.\\
\newline
$\blacktriangleright$ Proof, that  $\Omg$ is separating for $\Mg$:\\
 Let $A\in \Mana$.
We choose $\phi \in \bigcup_{n\in\N} \ran \one[|(\Lo+Q)|\,\leq\, n]$.
First, we have
\begin{equation*}
\Jg\, A^*\,\Omega(\beta/2)\,=\,\taug_{\imath \beta/2}(A)\,\Omega(\beta/2).
\end{equation*}
Let 
$f_\phi(z)=\langle \phi| \taug_z(A)\Omega(\beta/2)\rangle$
and 
$g_\phi(z)=\langle e^{\ovl{z}(\Lo+\Qg)}\phi\,
|\,A e^{-(\beta/2+z)(\Lo+\Qg)}\Omo\rangle$\\
for $-\beta/2\,\leq \,\Re z \,\leq \,0$. Both functions are continuous
and analytic if $-\beta/2\,<\, \Re z \,< 0$. Furthermore,
$f_\phi(\imath \,t)\,=\,g_\phi(\imath \,t)$ for $t\in \R$.
Hence $f_\phi\,=\,g_\phi$ and for $z\,=\,-\beta/2$
\begin{equation*}
\langle\, \phi\,|\,\Jg \,A^*\Omega(\beta/2)\,\rangle
\,=\,\langle\, e^{-\beta/2\,(\Lo+\Qg)}\phi\,|\,A\Omo\,\rangle.
\end{equation*}
This equation extends to all $A\in \Mg$, we
obtain $A\,\Omo \in \dom(e^{-\beta/2\,(\Lo+\Qg)})$, such as
$e^{-\beta/2(\Lo+\Qg)}\,A\,\Omo\,=\,\Jg\, A^*\Omega(\beta/2)$ for $A\in \Mg$.
Assume $A^*\,\Omega(\beta/2)\,=\,0$, then \\$e^{-\beta/2\,(\Lo+\Qg)}A\Omo\,=\,0$
and  hence $A\Omo\,=\,0$. Since $\Omo$ is separating, it follows
that $A=0$ and therefore $A^*\,=\,0$.\\
\newline
$\blacktriangleright$ Proof of $\Omg\in\mathcal{C}$, and that $\Omg$ is cyclic for $\Mg$:\\ 
To prove that $\phi\in \mathcal{C}$ it is sufficient to
check that $\langle\, \phi\,|\, A\Jg A\Omo\rangle \,\geq \,0$ for all $A\in\Mg$.
We have 
\begin{align*}
\langle\, \Omega(\beta/2)|\, A \Jg A\,\Omo\,\rangle\,
&=\, \ovl{\langle\,  \Jg A^*\Omega(\beta/2)\,| A\Omo\,\rangle}\\
\,&=\, \ovl{\langle\,  e^{-\beta/2(\Lo+\Qg)} \,A\Omo\,|\, A\Omo\,\rangle}\,\geq\, 0.
\end{align*}
The proof follows, since every separating element of $\mathcal{C}$ is cyclic.
\newline
$\blacktriangleright$ Proof, that $\omg$ is a $(\taug,\,\beta)$-KMS state:\\
For $A,\,B\in \Mg$ and $z\in S_{\beta}$ we define
\begin{equation*}
F_\beta(A,\,B,\,z)\,
=\, c\,\langle\, e^{-\imath \ovl{z}/2\Lg}A^*\Omg\,|\,e^{\imath z/2\Lg}B\Omg\,\rangle,
\end{equation*}
where $c\,:=\, \|\Omg\|^{-2}$. First, we observe
\begin{align*}
F_\beta(A,\,B,\,t)&\,=\,c\,\langle\, e^{-\imath t/2\Lg}A^*\Omg
\,|\,e^{\imath t/2\Lg}B\Omg\,\rangle
\,=\,c\,\langle\, \Omg\,|A\taug_t(B)\Omg\,\rangle\\ \nonumber
&\,=\,\omg(A\,\taug_t(B))
\end{align*}
and
\begin{align*}
\omg(\taug_t(B)A)&\,=\,c\,\langle\, \taug_t(B^*)\Omg\,|\,A\Omg\,\rangle
\,=\,c\,\langle\, \Jg A\Omg\,|\,\Jg \taug_t(B^*)\Omg\,\rangle \\ \nonumber
&\,=\,c\,\langle\, e^{-\beta/2 \Lg}A^*\Omg\,|\,e^{-\beta/2\, \Lg}\taug_t(B)\Omg\,\rangle\\ 
&\,=\,c\,\langle\, e^{-\imath \ovl{(\imath \beta+t)}/2 \Lg}A^*\Omg\,|
\,e^{\imath \,(\imath \beta+t)/2 \Lg}B\Omg\,\rangle\\ \nonumber
&\,=\,F_\beta(A,\,B,t+\imath \beta).
\end{align*}      
The  requirements on the analyticity of $F_\beta(A,\,B,\,\cdot)$ 
follow from Lemma \ref{Lem2App}.
\end{proof}
\bigskip
\section{ Proof of Theorem \ref{Thm2}}\label{Eqstate}
For $\unl{s}_{\,n} \,:=\, (s_n,\ldots,\,s_1)\in \R^n$ we define
\begin{equation} \label{DefQNs}
\QN(\unl{s}_{\,n})\,:= \,\QN(s_n)\cdots\QN(s_1),\qquad \QN(s)\,:=\, e^{-s \Lo}\QN e^{s\,\Lo},\ s\in \R
\end{equation}
At this point, we check that $\QN(\unl{s}_{\,n})\Omo$ is well defined,
and that it is an analytic vector of $\Lo$, see Equation \eqref{DefQN}.
The goal of Theorem \ref{Thm2}
is to give explicit conditions on $\Hael$ and $\We$,
which ensure $\Omo\in \dom(e^{-\beta/2\,(\Lo+\Qg)\,}).$
Let
\begin{eqnarray}\label{FormFactorBound}
\unl{\eta}_1 &\,:=\,& \int  \big(\|\vec G(k)\|^2_{\mathcal{B}(\Hel)}
+\|\vec H(k)\|^2_{\mathcal{B}(\Hel)}\big)(2+4\alpha(k)^{-1})\,dk\\ \nonumber
\unl{\eta}_2 &\,:=\,& \int  \big(\| F(k)\,\Haelp^{-\gamma}\|^2_{\mathcal{B}(\Hel)}
+\|F(k)^*\,\Haelp^{-\gamma} \|_{\mathcal{B}(\Hel)}\big)(2 +4\alpha(k)^{-1})\,dk
\end{eqnarray}
The idea of the proof is the following.
First, we expand $ e^{-\beta /2(\Lo +\QN)}e^{\Lo}$ in a Dyson-series, i.e.,
\begin{align}\label{DysonSer}
e&^{-\beta /2(\Lo +\QN)}e^{\Lo} \\\nonumber
&= \one +\sum_{n=1}^\infty(-1)^n\int_{\Delta_{\beta/2}^{n}}
  e^{-s_n \Lo}\QN e^{s_n\,\Lo}\cdots e^{-s_1 \Lo}\QN e^{s_1\,\Lo}\, d\underline{s}_{\,n}.
\end{align}
Under the assumptions of Theorem \ref{Thm2} we obtain an upper bound, uniform in $N$, for
\begin{align}\label{DysonExp}
\langle \Omo\,|&\,e^{-\beta (\Lo +\QN)}\Omo \rangle\\ \nonumber
&= 1 +\sum_{n=1}^\infty(-1)^n
\int_{\Delta_{\beta}^{n}}\langle \Omo\,|\,
  e^{-s_n \Lo}\QN e^{s_n\,\Lo}\cdots e^{-s_1 \Lo}\QN e^{s_1\,\Lo}\Omo \rangle\,  d\underline{s}_{\,n}.
\end{align}
This is proven in Lemma \ref{MainEstimate} below, which is the most important part of this section.
In Lemma \ref{DysExp} and Lemma \ref{Lem0.1} we deduce 
from the upper bound for \eqref{DysonExp} an upper bound for
$\|e^{-(\beta/2)(\Lo+\QN)}\Omo\|$, which is uniform in $N$. The proof 
of Theorem \ref{Thm2} follows now from Lemma \ref{Lem1.1},
where we show that $\Omo\in \dom(e^{-(\beta/2)(\Lo+\Qg)})$.
%
%
\begin{Lemma}\label{DysExp}
Assume 
\begin{equation*}
\limsup_{n\,\rightarrow \,\infty} \sup_{0\,\leq \,x\, \leq\, \beta/2}
\Big\| \int_{\Delta_x^n}\QN(\unl{s}_{\,n})\,d\unl{s}_{\,n}\Big\|^{1/n}\,<\,1.
\end{equation*}
for all $N\in \N$.
Then $\Omo\in \dom(e^{-x(\Lo+\QN)}),\ 0\,<\, x\,\leq\, \beta/2$ and
\begin{equation}\label{Lem:DysExp}
 e^{-x\,(\Lo+\QN)}\Omo\,=\,\Omo+\sum_{n=1}^\infty(-1)^n\int_{\Delta_{x}^{n}}
  \QN(\unl{s}_{\,n})\Omo\, d\underline{s}_{\,n}.
\end{equation}
In this context $\Delta_{x}^{n}\,=\, \{(s_1,\ldots,\,s_n)\in \R^n\,
:\, 0\,\leq \,s_n\,\leq\ldots \leq s_1\,\leq\, x\}$
is a simplex of dimension $n$ and sidelength $x$.  
\end{Lemma}
\begin{proof}
Let $\phi \in \ran \one[ |\Lo\,+\,\QN|\,\leq\, k]$ 
and $0\,\leq \,x\,\leq\, \beta/2$ be fixed.
An $m$-fold application of the 
fundamental theorem of calculus yields
\begin{eqnarray} \nonumber
\lefteqn{
\langle\, e^{-x\,(\Lo\,+\,\QN)}\,\phi\, | \, e^{x\,\Lo}\,\Omo\,\rangle 
= \Big\langle\, \,\phi \,|\,\Omo\,+\, \sum_{n=1}^m(-1)^n \int_{\Delta_{x}^{n}}\,
 \QN(\unl{s}_{\,n})\,\Omo \,d\underline{s}_{\, n}\,\Big\rangle }\\ \label{eq3.16}
&&+ (-1)^{m+1}\int_{\Delta_{x}^{m+1}}\, \big\langle\, e^{-s_{m+1}\,(\Lo\,+\,\QN)}\phi 
\,|\, e^{s_{m+1}\,\Lo}\,\QN(\unl{s}_{\,m+1})\,\Omo \,\big\rangle\,d\underline{s}_{\, m+1}.
\end{eqnarray}
Since $\Lo\,\Omo \,=\,0$ we have for $r(\unl{s}_{\,m+1})\,:=\,(s_m-s_{m+1},\ldots,\, s_1-s_{m+1})$
that
\begin{equation*}
e^{s_{m+1}\Lo}\,\QN(\unl{s}_{\,m+1})\Omo\,=\,\QN\, \QN(r(\unl{s}_{\,m+1}))\Omo,
\end{equation*}
We turn now to the second expression on the right side of Equation \eqref{eq3.16},
after a linear transformation depending on $s_{m+1}$ we get
\begin{equation*}
(-1)^{m+1}\,\int_0^x  \Big\langle\, e^{-s_{m+1}(\Lo+\QN)}\phi \,|\,\QN \,\int_{\Delta^m_{x-s_{m+1}}}
\QN(\unl{r}_{\,m})\Omo\, d\underline{r}_{\, m}\,\Big\rangle\,ds_{m+1}.
\end{equation*}
Since 
$\| e^{-s_{m+1}\,(\Lo\,+\,\QN)}\,\phi\|\,\leq \,e^{\beta/2\, k}\,\|\phi\|$, 
and using that $\QN(\unl{r}_{\,m})\,\Omo$ is a state with at most $2m$ bosons,
we obtain the upper bound
\begin{equation*}
\const\|\phi\|\, \sqrt{(2m)\,(2m+1)}\,\sup_{0\,\leq\, x\,\leq\, \beta/2} 
\Big \| \int_{\Delta^m_{x-s_{m+1}}}
\QN(\unl{r}_{\,m})\,\Omo\, d\underline{r}_{\, m}\Big \|.
\end{equation*}
Hence, for $m\,\rightarrow\,\infty$ we get
\begin{equation*}
\langle\, e^{-x(\Lo+\QN)}\phi\mid \Omo\,\rangle 
= \Big\langle\, \phi\,|\, \Omo\,+\,\sum_{n=1}^\infty(-1)^n \int_{\Delta_{x}^{n}}\,
 \QN(\unl{s}_{\,n})\,\Omo \,d\underline{s}_{\, n}\,\Big\rangle .
\end{equation*}
Since $\bigcup_{k=1}^\infty \ran \one[ |\Lo\,+\,\QN|\,\leq \,k]$ is a core of
$e^{-x(\Lo\,+\,\QN)}$, the proof follows from the 
self-adjointness of $e^{-x(\Lo\,+\,\QN)}$.
\end{proof}
%
%
%

\begin{Lemma}\label{Lem0.1}
Let $0\,<\,x\,\le\,\beta/2$. We have the identity
\begin{eqnarray}
\lefteqn{
\int_{\Delta_{x/2}^{n}}\, \int_{\Delta_{x/2}^{m}}
\big \langle\, \QN(\unl{r}_{\,m})\Omo\,|\,\QN(\unl{s}_{\,n})\Omo  \,\big\rangle\, d\underline{r}_{\,m}\,d\underline{s}_{\,n}}\\ \nonumber
&=&
\int_{\Delta_{\beta}^{n+m}}\, \mathbf{1}[z_m\,\geq \,\beta-x\,\geq\, x\, \geq\, z_{m+1}]\,
\big \langle\, \Omo\,\big|\,\QN(\unl{z}_{\,n+m})\Omo\,\big \rangle\,d\underline{z}_{n+m}.
\end{eqnarray}
For $m\,=\,n$ it follows
\begin{equation}
\Big\|\int_{\Delta_{x/2}^{n}}\QN(\unl{s}_{\,n})\Omo\, d\unl{s}_{\,n}\Big\|^2 
\leq \int_{\Delta_{\beta}^{2n}}\big|\big\langle\, \Omo\,|\,\QN(\unl{s}_{\,2n})\Omo\, 
\big\rangle \big|\,d\unl{s}_{\,2n}.
\end{equation}
\end{Lemma}
%
%
%
\begin{proof}
Recall Theorem \ref{FreeSystem} and Lemma \ref{Lem5d}.
Since $\Jg$ is a conjugation we have 
$\langle\, \phi\,|\, \psi \,\rangle =\langle\, \Jg \,\psi\,|\, \Jg\,\phi \,\rangle $,
and for every operator $X$, that is affiliated with $\Mg$, we have $\Jg\, X\,\Omo \,=\, e^{-\beta/2\,\Lo}\,X^*\,\Omo$.
Thus,
\begin{eqnarray} \label{eq3.17a}
\lefteqn{
\int_{\Delta_{x /2}^{n}}\, \int_{\Delta_{x/2}^{m}}
\big \langle\, \QN(\unl{r}_{\,m})\Omo\,|\,\QN(\unl{s}_{\,n})\,\Omo\, \big \rangle\, d\underline{r}_{\,m}\,d\underline{s}_{\,n}
}\\ \nonumber
&=& \int_{\Delta_{x/2}^{n}} \int_{\Delta_{x/2}^{m}}
\big \langle\, e^{-\beta/2\Lo}\,\QN(\unl{s}_{\,n})^*\,\Omo\,\big|\,e^{-\beta/2\Lo}\,\QN(\unl{r}_{\,m})^*\,\Omo\,\big \rangle\,
d\underline{r}_{\,m}\,d\underline{s}_{\,n}
\end{eqnarray}
Since $\Lo\, \Omo\,=\,0$ we have
\begin{equation*}
e^{-\beta\Lo}\,\QN(\unl{r}_{\,m})^*\,\Omo\, = \,\QN(\beta-r_1)\cdots \QN(\beta-r_m)\,\Omo.
\end{equation*}
Next, we introduce new variables for $\underline{r}$, namely
$y_i \,:=\, \beta \,-\,r_{m-i+1}$. Let $D_{x/2}^{m}\,:=\, \{ \underline{y}_{\,m}\in \R^m\,:\
\beta-x \,\leq\, y_m\,\leq \ldots \leq\, y_1\, \leq\, \beta\}$.  Thus the right side of Equation
\eqref{eq3.17a} equals
\begin{eqnarray*}
\lefteqn{
\int_{\Delta_{x/2}^{n}}\, 
\int_{D_{x/2}^{m}}\,
\big \langle\, \Omo\,\big|\,\QN(\unl{s}_{\,n})\,\QN(\unl{y}_{\,m})^*\,\Omo \,\big \rangle \,d\underline{s}_{\,n}\,d\underline{y}_{\,m}
}\\
&=&
\int_{\Delta_{\beta}^{n+m}}\, \one[z_m\,\geq\, \beta\,-\,x\,\geq\, x\, \geq\, z_{m+1}]\,
\big \langle\, \Omo\,\big|\,\QN(\unl{z}_{\,n+m})\,\Omo \,\big \rangle\,d\underline{z}_{n+m}.
\end{eqnarray*}
The second statement of the Lemma follows by choosing $n\,=\,m$.
\end{proof}
%
%
\begin{Lemma}\label{Lem1.1}
Assume $\sup_{N\in \N}\|e^{-x(\Lo+\QN)}\Omo\|\,<\,\infty$
then $\Omo\in \dom(e^{-x(\Lo+\Qg)})$ and
\begin{equation*}
\|e^{-x(\Lo+\Qg)}\Omo\|\,\leq\,\sup_{N\in \N}\|e^{-x(\Lo+\QN)}\,\Omo\|
\end{equation*}
\end{Lemma}
%
%
\begin{proof}
\indent For $f\in \mathcal{C}^\infty_0(\R)$ and $\phi\in\Kg$
we define $\psi_N\,:=\,f(\Lo\,+\,\QN)\,\phi$. Obviously, 
for $g(r)\,=\,e^{-x\,r}\,f(r)\in \mathcal{C}^\infty_0(\R)$ we have
$e^{-x\,(\Lo\,+\,\QN)}\,\psi_N\,=\,g(\Lo\,+\,\QN)\,\phi $. 
Since $\Lo\,+\,\QN$ tends to $\Lo\,+\,\Qg$ 
in the strong resolvent sense as $N\to \infty$, 
we know from \cite{ReedSimonI1980} that
$\lim_{N\,\rightarrow\, \infty}\psi_N\,=\,f(\Lo\,+\,\Qg)\,\phi\,=:\, \psi$
and
\begin{equation*}
\lim_{N\rightarrow \infty}e^{-x\,(\Lo\,+\,\QN)}\,\psi_{N}
\,=\,\lim_{N\rightarrow \infty}g(\Lo\,+\,\QN)\,\phi
\,=\,g(\Lo\,+\,\Qg)\,\phi\,=\, e^{-x\,(\Lo\,+\,\Qg)}\,\psi.
\end{equation*}
Thus,
\begin{align*}
|\langle\, e^{-x(\Lo+\Qg)}\,\psi\,|\,\Omo\,\rangle |
\,&=\, \lim_{N\to \infty} |\langle\, e^{-x(\Lo+\QN)}\psi_N \,|\,\Omo\,\rangle |\\
\,&\leq\, \sup_{N\in \N}\|e^{-x\,(\Lo+\QN)}\Omo\|\, \|\psi\|,
\end{align*}
Since $\{ f(\Lo+\Qg)\,\phi\in\Kg\,:\,\phi \in \Kg,\, f\in \mathcal{C}^\infty_0(\R)\}$
is a core of $e^{-x(\Lo+\Qg)}$, we obtain $\Omo\in \dom(e^{-x(\Lo+\Qg)})$.
\end{proof}
%
%
\begin{Lemma}\label{MainEstimate}
For some $C>0$ we have 
\begin{align*}
\int_{\Delta_{\beta}^n}
&\Big|\big\langle\, \Omo\,|\, \QN(\unl{s}_{\,n})\Omo\,\big\rangle \Big|\,d\unl{s}_{\,n}\\ \nonumber
&\leq \const \,(n+1)^2\,(1+\beta)^n\,
\Big(8\unl{\eta}_1 +\frac{(8C\unl{\eta}_2)^{1/2} }{(n+1)^{(1-2\gamma)/2}} \Big)^{n},
\end{align*}
where $\unl{\eta}_1$ and $\unl{\eta}_2$ are defined in \eqref{FormFactorBound}.
\end{Lemma}
%
%
\begin{proof}[Proof of \ref{MainEstimate}]
First recall the definition of $\QN$ and $\QN(\unl{s}_{\,n})$ in
Equation \eqref{DefQN} and Equation \eqref{DefQNs}, respectively. Let
\begin{equation*}
\int_{\Delta_{\beta }^{n}}\big|\big\langle\, \Omo
\,|\, \QN(\unl{s}_{\,n})\,\Omo\,\big\rangle \big|\,d\unl{s}_{\,n}
\,=:\, \int_{\Delta_{1}^{n}}\beta^n\,J_n(\beta,\unl{s})\,d\unl{s}_{\,n},
\end{equation*}
The functions $J_n(\beta,\unl{s})$ clearly depends on $N$,
but since we want to find an upper bound independent of $N$,
we drop this index. Let $W_1=\Phi(\vec G )\,\Phi(\vec H )\,+\,\hc\,$,
$W_2:=\Phi(F)$ and $W:=W_1+W_2$.
By definition of $\omo$ in \eqref{FreeSystem}, see also \eqref{AWRep}, 
we obtain
\begin{align*}
J_{\,n}(\beta,\,\unl{s}_{\,n})\,&=\, 
\omo\Big(\big(e^{-\beta\, s_n\,\Ho}\,W\,e^{\beta s_n\Ho}\big)
\cdots\big(e^{-\beta \,s_1\,\Ho}\,W\,e^{\beta \,s_1\,\Ho}\big)\Big)\\
&=(\Zel)^{-1}\sum_{\kappa\in \{ 1,\,2\}^n}
\omf\Big(\Tr_{\Hel}\big\{e^{-\beta\, \Hael}\big(e^{-\beta \,s_n\,\Ho}\,W_{\kappa(n)}
\,e^{\beta\, s_n\,\Ho}\big)\cdots\\
&\phantom{(\Zel)^{-1}\sum_{\kappa\in \{ 1,\,2\}^n}
\omf\Big(\Tr_{\Hel}\big\{e^{-\beta\, \Hael}\ }
\cdots\big(e^{-\beta\, s_1\,\Ho}\,W_{\kappa(1)}\,e^{\beta\, s_1\,\Ho}\big)\big\}\Big)
\end{align*}
By definition of $\omf$ it suffices to consider expressions with an even
number of field operators. In the next step we sum over all
expression, where $n_1$ times $W_1$ occurs and $2n_2$ times $W_2$.
The sum of $n_1$ and $n_2$ is denoted by $m$. For fixed $n_1$ and $n_2$
the remaining expressions are all expectations in $\omf$ of $2m$ field operators.
In this case the expectations in $\omf$ can be expressed by an integral over
$\R^{2m}\times \{\pm\}^{2m}$ with respect to $\nu$, which is defined in Lemma \ref{Lem-N1}
below.
To give a precise formula we define
$$M(m_1,\,m_2)\,=\, \{ \kappa\in \{1,\,2\}^n\, 
:\,\#\kappa^{-1}(\{i\})\,=\,m_i,\quad i=1,\,2\}.$$
Thus we obtain
\begin{align}\label{eq3.41}
J_n&(\beta,\,\unl{s}_{\,n})
\,=\,(\Zel)^{-1}\sum_{\stackrel{(n_1,\,n_2)\in \N^2}{n_1\,+\,2n_2\,=\,n}}\qquad
\sum_{\stackrel{\kappa\in M(n_1,\,2n_2)}{m\,:=\,n_1\,+\,n_2}}
\int\,\nu(d\unl{k}_{\,2m}\tensor d\unl{\tau}_{\,2m})\\ \nonumber
&\phantom{=}  
  \Tr_{\Hel}\big\{ e^{-(\beta\,-\,\beta(s_1\,-\,s_{2m}))\Hael}I_{2m}
e^{-\beta\,(s_{2m-1}\,-\,s_{2m})\,\Hael}
\cdots
e^{-\beta\,(s_1\,-\,s_{2})\,\Hael}\,I_{1}\big\}\,,
\end{align}
Of course $I_{j}$ depends on $\unl{k}_{\,2m}\times \unl{\tau}_{\,2m}$,
namely for $\kappa(j)\,=\,1,\,2$ we have
\begin{eqnarray*}
I_j 
&\,=\,&
\begin{cases}
I_j(m,\,\tau,\,m',\,\tau'),& \kappa(j)\,=\,1\\
I_j(m,\,\tau),& \kappa(j)\,=\,2,
\end{cases} 
\end{eqnarray*}
where $(m,\,\tau),\, ( m',\,\tau')\in \{ (k_j,\tau_j)\,:\, j=1,\ldots,m\}$.
For $\kappa(j)\,=\,1$ we have that
\begin{eqnarray*}
I_j(m,\,+,\,m',\,-)&\,=\,& \vec{G}^*(m)\, \vec{H}(m')\,+\,\vec{H}^*(m)\, \vec{G}(m')\\ \nonumber
I_j(m,\,-,\,m',\,+)&\,=\,& \vec{G}(m)\, \vec{H}^*(m')\,+\,\vec{H}(m)\, \vec{G}^*(m')\\ \nonumber
I_j(m,\,+,\,m',\,+)&\,=\,& \vec{G}^*(m)\, \vec{H}^*(m')\,+\,\vec{H}^*(m)\, \vec{G}^*(m')\\ \nonumber
I_j(m,\,-,\,m',\,-)&\,=\,& \vec{G}(m)\, \vec{H}(m')\,+\,\vec{H}(m)\, \vec{G}(m')
\end{eqnarray*}
and for $\kappa(j)\,=\,2$ we have that
\begin{eqnarray*}
I_j(m,\,+)&\,=\,& F^*(m)\\ \nonumber
I_j(m,\,-)&\,=\,& F(m).
\end{eqnarray*}
In the integral \eqref{eq3.41} we insert for $(m,\,\tau)$ and $(m',\,\tau')$ 
in the definition of $I_j$ from left to right $k_{2m},\,\tau_{2m},\ldots,\,k_1,\,\tau_1$.\\
\indent For fixed $(\unl{k}_{\,2m},\unl{\tau}_{\,2m})$ the
integrand of  \eqref{eq3.41} is a trace of a product of $4m$
operators in $\Hel$. We will apply Hölder's-inequality for the trace, i.e.,
\begin{equation*}
|\Tr_{\Hel}\{A_{2m}\,B_{2m}\cdots A_1\,B_1\}|\leq \prod_{j=1}^{2m}\|B_{j}\|_{\mathcal{B}(\Hel)}
\cdot \prod_{j=1}^{2m}\Tr_{\Hel}\{A_i^{p_j}\}^{p_j^{-1}}.
\end{equation*}
In our case $p_i\,:=\,(s_{i-1}\,-\,s_i)^{-1}$ for $i=2,\ldots,\,2m$ 
and $p_1\,:=\,(1\,-\,s_1\,+\,s_{2m})^{-1}$ and
\begin{equation*}
(A_j,\, B_j)\,:=\,
\begin{cases}
\big( e^{-\beta\, p_j^{-1}\, \Hael},\, I_j(m,\,\tau,\,m',\,\tau')\big),& \kappa(j)\,=\,1\\
\big( e^{-\beta\, p_j^{-1}\, \Hael}\,\Haelp^\gamma,\, \Haelp^{-\gamma}\,I_j(m,\tau,)\big),& \kappa(j)\,=\,2
\end{cases}.
\end{equation*}
We define
\begin{align*}
\eta_1(k) \,&=\,\max\big\{ \|\vec G(k)\|_{\mathcal{B}(\Hel)^r},\, \|\vec H(k)\|_{\mathcal{B}(\Hel)^r} \big\}\\
\eta_2(k) \,&=\,\max\big\{\|F(k)\,\Haelp^{-\gamma}\|_{\mathcal{B}(\Hel)},
\,\|F^*(k)\,\Haelp^{-\gamma}\|_{\mathcal{B}(\Hel)}\big\}.
\end{align*}
By definition of $B_j$ we have
\begin{equation} \label{UpperBoundB}
\|B_{j}\|_{\mathcal{B}(\Hel)}\le
\begin{cases} 
\eta_1(m)\eta_1(m'),& \kappa(j)=1\\
\eta_2(m),& \kappa(j)=2
\end{cases}.
\end{equation}
Furthermore,
\begin{align*}
\Tr_{\Hel}\{A_i^{p_j}\}^{p_j^{-1}}&=\Tr_{\Hel}\big\{ e^{-\beta \Hael}\,\Haelp^{p_j\gamma}\big\}^{p_j^{-1}}\\
&\leq \| e^{-\epsilon \Hael}\,\Haelp^{p_j\gamma}\|^{p_j^{-1}}_{\Hel}\, 
\Tr_{\Hel}\big\{ e^{-(\beta-\epsilon)\, \Hael}\big\}^{p_j^{-1}},\quad k(j)=2
\end{align*}
Let $E_{gs}\,:=\,\inf\, \sigma(\Hael)$. The spectral theorem for self-adjoint operators implies
\begin{equation*}
 \| e^{-\epsilon\, \Hael}\,\Haelp^{p_i\,\gamma}\|^{p_i^{-1}}_{\Hel} 
 \,\leq\, \sup_{r\,\geq \,E_{gs}}\, e^{-\epsilon\, p_i^{-1} \,r}(r\,-\,E_{gs}\,+1)^{\gamma}
   \,\le\,  \epsilon^{-\gamma}\,p_i^\gamma\, e^{-\epsilon \,p_i^{-1}\,(E_{gs}\,-\,1)}.
\end{equation*}
Inserting this estimates we get 
\begin{align*}
  \Tr_{\Hel}&\big\{ e^{-(\beta\,-\,\beta(s_1\,-\,s_{2m}))\Hael}I_{2m}
e^{-\beta\,(s_{2m-1}\,-\,s_{2m})\,\Hael}
\cdots
e^{-\beta\,(s_1\,-\,s_{2})\,\Hael}\,I_{1}\big\}\\
&\le C_{\kappa}(\unl{s}_{n}) 
\prod_{j=1}^{2m}\|B_j\|_{\mathcal{B}(\Hel)}
\end{align*}
where
\begin{equation}
 C_{\kappa}(\unl{s}_{n}):=(1-\,s_1\,+\,s_{n})^{-\alpha_1}
\prod_{i=1}^{n-1}(s_i\,-\,s_{i+1})^{-\alpha_i}
\end{equation}
and
\begin{equation}
\alpha_i=\begin{cases}
0,&\kappa(i)=1\\
1/2,&\kappa(i)=2
\end{cases}
\end{equation}
Now, we recall the definition of $\nu$. Roughly speaking,
one picks a pair of  variables  $(k_i,k_j)$
and integrates over 
$\delta_{k_i,k_j}\coth(\beta/2\alpha(k_i))\,dk_i dk_j$.
Subsequently one picks the next pair and so on. At the end
one sums up all $\frac{(2m)!}{2^m\, m!}$ pairings and 
all $4^m$ combinations of $\unl{\tau}_{\,2m}$. 
Inserting Estimate \eqref{UpperBoundB} and
that
\begin{align*}
\int &\eta_\nu(k)\eta_{\nu'}(k)\coth(\beta/2\alpha(k))\,dk \leq (1+\beta^{-1})\unl{\eta}_{\,\nu}^{1/2}\unl{\eta}_{\,\nu'}^{1/2},
\end{align*}
we obtain
\begin{align*}
|J_n(\beta,\,\unl{s})|
&\leq \frac{(1+\beta^{-1})^n}{\Zel} \sum_{\stackrel{(n_1,\,n_2)\in \N_0^2}{n_1\,+\,2n_2\,=\,n}}
\ \sum_{\stackrel{\kappa\in M(n_1,\,2n_2)}{m\,:=\,n_1\,+\,n_2}} 
(\,\unl{\eta}_1)^{n_1}\,(C\unl{\eta}_2)^{n_2}\,\frac{(2m)!2^m}{ m!}C_\kappa(\unl{s})
\end{align*}
By Lemma \ref{Lem0.2} below and since $(2m)!/(m!)^2\le 4^m$ we have
\begin{align*}
\int_{\Delta_{\beta }^{n}}&\big|\big\langle\, \Omo
\,|\, \QN(\unl{s}_{\,n})\,\Omo\,\big\rangle \big|\,d\unl{s}_{\,n}\\
&\,\le \const (1+\beta)^n \sum_{\stackrel{(n_1,\,n_2)\in \N_0^2}{n_1\,+\,2n_2\,=\,n}}
     { n \choose n_1}  \frac{(8\unl{\eta}_1)^{n_1}\,(8C'\unl{\eta}_2)^{n_2}}{(n+1)^{(1-2\gamma)\,n_2 -2} }  
\end{align*}
This completes the proof.
\end{proof}
%
%
%
%
\section{The Harmonic Oscillator}\label{HarmOsc}
Let $L^2(X,\,d\mu)=L^2(\R)$ and $\Hael=: \Hosc:= -\Delta_q+\Theta^2 q^2$
be the one dimensional harmonic oscillator and $\hph=L^2(\R^3)$. We define
\begin{equation}
\Hg\,=\, \Hosc \,+\, \Phi(F) \,+\, \Haf,\qquad \Haf \,:= \, d\Gamma(|k|),
\end{equation}
where $\Phi(F)\,=\, q\, \cdot \,\Phi(f)$, with $\lambda\,(|k|^{-1/2}\,+\,|k|^{1/2})\,f\in L^2(\R^3)$.\\
\indent $\Hosc$ is the harmonic oscillator, the form-factor $F$ comes from
the dipole approximation.\\
The Standard Liouvillean for this model is denoted by $\mathcal{L}_{osc}$.
Now we prove Theorem \ref{Thm4a}.
\begin{proof}
We define the creation and annihilation operators for the electron.
\begin{gather}\label{eq4.1}
\be^*\,=\,\frac{\Theta^{1/2}\,q\,-\,\imath \, \Theta^{-1/2}\, p}{\sqrt{2}},
\quad \be\,=\,\frac{\Theta^{1/2}\,q\,+\,\imath\, \Theta^{-1/2}\,p}{\sqrt{2}},\quad p\,=\,-\imath \,\partial_x,\\
\Phi(c)\,=\, c_1\,q\,+\,c_2 \, p,\quad \textrm{for}\ c\,=\,c_1\,+\,\imath \,c_2\in \C,\ c_i\in \R.
\end{gather} %
These operators fulfill the CCR-relations and the harmonic-
oscillator is the number-operator up to constants.
\begin{gather}\label{eq4.2}
[\be,\,\be^*]\,=\,1,\qquad [\be^*,\,\be^*]\,=\,[\be,\,\be]\,=\,0,\qquad \Hosc\,=\,\Theta \be^*\,\be\,+\,\Theta/2,\\
[\Hosc,\,\be]\,=\,-\Theta \,\be,\qquad [\Hosc,\,\be^*]\,=\,\Theta \,\be^*.
\end{gather}
The vector $\Omega\,:=\, \left(\frac{\Theta}{\pi}\right)^{1/4} \,e^{-\Theta\, q^2\,/2}$ 
is called the vacuum vector. Note, that one can identify $\mathcal{F}_b[\C]$ with $L^2(\R)$, 
since $\linhull\{ (A^*)^n\,\Omega\, | \, n\in \N^0\}$ is dense
in $L^2(\R)$. It follows, that $\omosc$ is quasi-free, as a state over $\We(\C)$ and
\begin{equation}
\omosc(\We(c))\,=\, (\Zel)^{-1}\, \Tr_{\Hel}\{ e^{-\beta \Hael} \,\We(c)\}
\,=\, \exp\big( -1/4\, \coth(\beta\, \Theta/2)\,|c|^2\big),
\end{equation}
where $\Zel \,=\, \Tr_{\Hel}\{ e^{-\beta\,\Hel} \}$ is the partition function for $\Hel$.\\
\indent First, we remark, that Equation \eqref{DefQNs} is defined for this
model without regularization by $P_N \,:=\, \one [\Hael\,\leq\, N]$.
Moreover we obtain from Lemma \ref{Lem0.1}, that
\begin{equation}
\Big\|\int_{\Delta_{\beta/2}^{n}}\Qg(\unl{s}_{\,n})\Omo \,d\unl{s}_{\,2n}\Big\|^2
\,\leq \,
\int_{\Delta_{\beta}^{2n}}\big|\big\langle\, \Omo\,|\, 
\Qg(\unl{s}_{\,2n})\,\Omo\,\big\rangle \big|\,d\unl{s}_{\,2n}\,=:\,
h_{2n}(\beta,\, \lambda).
\end{equation}
To show that $\Omg \in \dom(e^{-\beta/2 \,(\Lo\,+\,\Qg)})$ is suffices to
prove, that $\sum_{n\,=\,0}^\infty h_{2n}(\beta,\,\lambda)^{1/2}\,<\,\infty$.
We have
\begin{gather}\label{eq4.3}
h_{2n}(\beta,\, \lambda)\,=\, \frac{(-\beta\,\lambda)^{2n}}{\Zel} \int_{\Delta_1^{2n}}\, 
\omosc\big( \big( e^{-\beta \,s_{2n}\, \Hael}\,q\,e^{\beta\, s_{2n}\, \Hael}\big) \\ \nonumber
\cdots \big( e^{-\beta\, s_1\,\Hael}\,q\,e^{\beta \,s_1\,\Hael}\big)\big)\\ \nonumber
\cdot  \omf\big((e^{-\beta\, s_{2n}\, \Haf}\,\Phi(f)\,e^{\beta \,s_{2n} \,\Haf})
\cdots (e^{-\beta \,s_1\, \Haf}\,\Phi(f)\,e^{\beta \,s_1\, \Haf})\big)\, d\underline{s}_{\,2n}.
\end{gather}
Moreover, we have
\begin{align}\nonumber\label{eq4.4y}
e^{-\beta\, s_i \,\Hael} \,q\, e^{\beta\, s_i \,\Hael}
&= (2\Theta)^{-1/2}\big(e^{-\beta\, \Theta\, s_i}\,A^*\,+\,e^{\beta \,\Theta \,s_i}\,A\,\big)\\
e^{-\beta\, s_i \,\Haf}\,\Phi(f)\,e^{\beta\, s_i\, \Haf}&= 2^{-1/2}
\Big( a^*(e^{-\beta\, s_i\,|k|}\,f)+ a(e^{\beta\, s_i\,|k|}\,f)\Big).
\end{align}
Inserting the identities of Equation \eqref{eq4.4y} in Equation \eqref{eq4.3}
and applying Wick's theorem \cite[p. 40]{BratteliRobinson1996} yields
\begin{align}\nonumber
h_{2n}& (\beta,\,\lambda)
= (\beta\,\lambda)^{2n}  \int_{\Delta_1^{2n}}
\sum_{P\in \mathcal{Z}_2}\ \prod_{\{i,\,j\}\in P} 
K_{osc}(|s_i\,-\,s_j|,\,\beta)\\ \nonumber
&\phantom{\sum_{P\in \mathcal{Z}_2}\ \prod_{\{i,\,j\}\in P} 
K_{osc}(|s_i\,-}
\cdot\sum_{P'\in \mathcal{Z}_2}\ \prod_{\{k,\,l\}\in P'} 
K_{f}(|s_k-s_l|,\,\beta)\, d\unl{s}_{\,2n}\\ \label{eq4.6}
&=\frac{(\beta\,\lambda)^{2n}}{(2n)!} \int_{[0,\,1]^{2n}}
\sum_{P,\,P'\in \mathcal{Z}_2} \prod_{\stackrel{\{i,\,j\}\in P}{\{k,\,l\}\in P'}} 
 K_{osc}(|s_i-s_j|,\,\beta)\,K_{f}(|s_k-\,s_l|,\beta)\, d\unl{s}_{\,2n},
\end{align}
where for $k\,<\,l$ and $i\,<\,j$, such as 
\begin{eqnarray*}
K_{f}(|s_k\,-\,s_l|,\,\beta)
&\,:= \,&\omf((e^{-\beta\, s_k \,\Haf}\,\Phi(f)\,e^{\beta \,s_k \,\Haf})\,
(e^{-\beta \,s_l \,\Haf}\,\Phi(f)\,e^{\beta \,s_l\, \Haf}))\\
K_{osc}(|s_i\,-\,s_j|,\,\beta)
&\,:=\,&\omosc(e^{-\beta\, s_i\, \Hael}\,q\,e^{\beta \,s_i \,\Hael}\,
e^{-\beta \,s_j \,\Hael}\,q\,e^{\beta \,s_j\, \Hael}). 
\end{eqnarray*}
The last equality in \eqref{eq4.6} holds, 
since the integrand is invariant
with respect to a change of the axis of coordinates.\\
\indent We interpret two pairings $P$ and $P'\in \mathcal{Z}_2$ as an
indirected graph $G\,=\,G(P,\,P')$, where $M_{2n}\,=\,\{1,\ldots,\,2n\}$
is the set of points. Any graph in $G$ has two kinds of lines,
namely lines in $L_{osc}(G)$, which belong to elements
of $P$ and lines in $L_{f}(G)$, which belong to elements
of $P'$.\\
\indent Let $\mathcal{G}(A)$ be the set of undirected
graphs with points in $A\subset M_{2n}$, such that
for each point "i" in $A$, there is exact one
line in $L_f(G)$, which begins in "i", and
exact one line in $L_{osc}(G)$, which begins
with "i". $\mathcal{G}_c(A)$ is the set of
connected graphs. We do not distinguish, if
points are connected by lines in $L_f(G)$ or by
lines in $L_{osc}(G)$.\\
Let
\begin{align*}
\mathcal{P}_k\, :=\, \Big\{ P\,:&\, P=\{A_1,\ldots,\, A_k \},\ 
\emptyset\not=A_i\subset M_{2n},\\
\,
& A_i\cap A_j \,=\,\emptyset \textrm{ for } i\,\not=\,j,\,
\bigcup_{i\,=\,1}^k\, A_i\,=\, M_{2n}\Big\}
\end{align*}
be the family of  decompositions of $M_{2n}$ in $k$ 
disjoint set. It follows
\begin{align}\nonumber\label{eq4.7}
h_{2n}(\beta,\, \lambda)&\,=\,
\frac{(\beta\,\lambda)^{2n}}{(2n)!}\,\sum_{G\in \mathcal{G}(M_{2n})} \int_{M_{2n}}
 \prod_{\stackrel{\{i,j\}\in L_{osc}(G)}{\{k,\,l\}\in L_{f}(G)}} 
 K_{osc}(|s_i\,-\,s_j|,\,\beta)\\ \nonumber
&\phantom{\,=\,} 
K_{f}(|s_k\,-\,s_l|,\,\beta)\, d\unl{s}_{\,n}\\ \nonumber
&=\frac{(\beta\lambda)^{2n}}{(2n)!}\ \sum_{k\,=\,1}^{2n}\ \sum_{ \{A_1,\ldots,\, A_k\}\in \mathcal{P}_k }
 \ \sum_{\stackrel{(G_1,\ldots,\,G_k)}{G_a\in  \mathcal{G}_c(A_a)}}\prod_{a\,=\,1}^k J(G_a,\,A_a,\,\beta)\\ 
&=\frac{(\beta\lambda)^{2n}}{(2n)!} \ \sum_{k\,=\,1}^{2n}\ \frac{1}{k!} \ 
\sum_{ \stackrel{A_1,\ldots,\,A_k\subset M_{2n},}{ \{A_1,\ldots,\, A_k\}\in \mathcal{P}_k} }\ 
\sum_{\stackrel{(G_1,\ldots,\,G_k)}{G_a\in  \mathcal{G}_c(A_a)}}\prod_{a\,=\,1}^k \,J(G_a,\,A_a,\,\beta),               
\end{align}
where
\begin{equation}\label{eq4.8}
 J(G_a,\,A_a,\,\beta)\,:=\,\int_{A_a}\prod_{\stackrel{\{i,\,j\}\in L_{osc}(G_a)}{\{k,l\}\in L_{f}(G_a)}} 
 K_{osc}(|s_i\,-\,s_j|,\,\beta)  K_{f}(|s_k\,-\,s_l|,\,\beta)\,d\underline{s}.
\end{equation}
$\int_{A_a}\,d\unl{s}$ means, $\int_{-1}^1\,ds_{j_1}\int_{-1}^1\,ds_{j_2}\ldots \int_{-1}^1\,ds_{j_m}$, where
$A_a\,=\,\{ j_1,\ldots,\,j_m\}$ and  $\#A_a\,=\,m$.\\
\indent From the first to the second line we summarize terms with graphs,
having connected components containing the same set of points.
From the second to the third line the order of the components
is respected, hence the correction factor $\frac{1}{k!}$ is introduced.
Due to Lemma \ref{Lem1} the integral depends only on
the number of points in the connected graph, i. e. 
$J(G,\,A,\,\beta)\,=\, J(\#A,\,\beta)$. Moreover, Lemma \ref{Lem1}
states that 
$\beta^{\#A}\cdot J(\#A,\,\beta)\,
\leq\,(2\||k|^{-1/2}\,f\|_2 \,(\Theta\, \beta)^{-1})^{\#A}\,(C\,\beta\,+\,1).$
To ensure that $\mathcal{G}_c(A_a)$ is not empty, $\#A_a$ must be even.
For $(m_1,\ldots,\,m_k)\in \N^k$ with $m_1\,+\cdots+\,m_k\,=\,n$ we obtain
\begin{equation}\label{eq4.9}
\sum_{ \stackrel{A_1,\ldots,\,A_k\subset M_{2n}, \#A_i\,=\,2m_i}
{ \{A_1,\ldots,\, A_k\}\in \mathcal{P}_k}} 1 
\,=\, \frac{(2n)!}{(2m_1)!\cdots(2m_k)!}.
\end{equation}
\indent Let now be $A_a\subset M_{2n}$ with $\#A_a\,=\,2m_a\,>\,2$ fixed. 
In $G_a$ are $\#A_a$ lines in $L_{osc}(G_a)$, since such lines have
no points in common, we have $\frac{(2m_a)!}{m_a!\, 2^{m_a}}$ choices. Let now
be the lines in $L_{osc}(G_a)$ fixed. We have now 
$\big((2m_a\,-\,2)(2m_a\,-\,4)\cdots 1\big)$ choices for
$m_a$ lines in $L_{f}(G_a)$, which yield a connected graph. Thus
\begin{equation}\label{eq4.10}
\sum_{G_a\in  \mathcal{G}_c(A_a),}\, 1
\,=\, \frac{(2m_a)!}{m_a!\, 2^{m_a}}\big((2m_a\,-\,2)\,(2m_a-\,\,4)\cdots 1\big)
\,=\, \frac{(2m_a)!}{2m_a }.
\end{equation}
For $\#A_a\,=\,2$ exists only one connected graph. We obtain
for $h_{2n}$
\begin{eqnarray}\label{eq4.11}
\lefteqn{
h_{2n}(\beta,\, \lambda)\,=\,(\lambda)^{2n} 
\sum_{k\,=\,1}^{2n}\frac{1}{k!} \sum_{ \stackrel{(m_1,\ldots,\,m_k)\in \N^k}{m_1\,+\ldots+\,m_k\,=\,n}}
\prod_{a\,=\,1}^k \frac{J(2m_a,\,\beta)(\beta^2)^{m_a}}{2m_a} }\\ \nonumber
 &\leq&(2\Theta^{-1}\, \||k|^{-1/2}\,f\|\,\lambda)^{2n} 
 \sum_{k\,=\,1}^{2n}\frac{1}{k!} \sum_{ \stackrel{(m_1,\ldots,\,m_k)\in \N^k}{m_1\,+\ldots+\,m_k\,=\,n}}
 \prod_{a\,=\,1}^k \frac{(C\,\beta\,+\,1)}{2m_a}\\ \nonumber
 &\leq&(2\Theta^{-1}\,\||k|^{-1/2}\,f\|\, \lambda)^{2n} \sum_{k\,=\,1}^{2n}
 \frac{\big((C\,\beta\,+\,1)/2\sum_{m\,=\,1}^n\frac{1}{m}\big)^k}{k!}.
\end{eqnarray} 
Since the $\sum_{m=1}^n\frac{1}{m}$ can be considered as a lower Riemann sum
for the integral $\int_1^{m\,+\,1} r^{-1}\,dr$,
we have $\sum_{m\,=\,1}^n\frac{1}{m}\,\leq \,\ln(n+1)$. Thus,
\begin{eqnarray}\label{eq4.11a}
h_{2n}(\beta,\, \lambda)
 &\leq& (2\Theta^{-1} \,\||k|^{-1/2}\,f\|\, \lambda)^{2n} 
 \sum_{k=1}^{2n}\frac{\big((C\,\beta\,+\,1)/2\,\ln(n+1)\big)^k}{k!}\\ \nonumber
 &\leq & (2\Theta^{-1}\, \||k|^{-1/2}\,f\|\, \lambda)^{2n} (n\,+\,1)^{(C\,\beta\,+\,1)/2}.
\end{eqnarray}
Since  $2 |\lambda|\, \||k|^{1/2}\,f\|\,<\,\Theta$ the series 
$\sum_{n=0}^\infty h_{2n}(\beta,\, \lambda)^{1/2}$ 
converges absolutely for all $\beta\,>\,0$.
It follows, that 
$$e^{-\beta/2 \,(\Lo\,+\,\Qg)}\,\Omo\,=\, \Omo\,+\,
\sum_{n\,=\,1}^\infty \int_{\Delta_{\beta/2}^{n}}\Qg(\unl{s}_{\,n})\,\Omo \,d\unl{s}_{\,n} $$
exists.  
\end{proof}
Conversely, Equation \eqref{eq4.11} and Lemma \ref{Lem1} imply
\begin{equation}\label{eq4.12}
h_{2n}(\beta,\, \lambda)\,\geq\,(\lambda/2)^{2n}  
  \frac{J(2n,\,\beta)\,\beta^{2n}}{2n} 
\,=\, \frac{\Big(\Theta^{-1}\,\int \frac{\beta^2\,\lambda^2/4\,|f(k)|^2}
{\sinh(|k|\,\beta/2) \sinh(\beta\,\Theta /2)}\,dk\Big)^{n}}{2n}.
\end{equation}
Hence for every $\beta\,>\,0$ exists a $\lambda\in\R$, such that
$h_{2n}(\beta,\, \lambda)\,\geq \,\frac{1}{2n}$. Thus $\sum_{n=1}^\infty h_{2n}(\beta,\,\lambda)^{1/2}\,=\,\infty$

\begin{Bem} We can therefore not extended Theorem \ref{Thm4}
to an existence proof for all $\lambda\,>\,0$.
\end{Bem}
%
%
\begin{Lemma}\label{Lem1} Following statements are true.
\begin{align*}
J(G,\,A,\,\beta)&= J(\#A,\beta),\ G\in \mathcal{G}_c(A)\\
J(\#A,\,\beta)&\leq (2\||k|^{-1/2}\,f\|_2 \,(\Theta\, \beta)^{-1})^{\#A}\cdot(C\,\beta\,+\,1)\\
J(\#A,\,\beta) &\geq \Big(\Theta^{-1}\,
\int \frac{|f(k)|^2}{\sinh(|k|\,\beta/2) \sinh(\Theta\, \beta/2)}\,dk\Big)^{\#A/2},
\end{align*}
where $\#A\,=\,2m$ and $C\,=\, (1/2)\,\frac{ \|f\|^2}{\||k|^{1/2}\,f\|^2}$.
\end{Lemma}%
%
%
\begin{proof}[Proof of \ref{Lem1}]
A relabeling of the integration variables yields
\begin{align*}
J(G,\,A,\,\beta)
\leq\,& \ovl{K}_f\,\int_{[0,1]^{2m}}\, 
K_{osc}(|t_1-t_2|,\,\beta)\,K_{f}(|t_2\,-\,t_3|,\,\beta)\cdots\\ \nonumber
&\cdots K_{osc}(|t_{2m-1}\,-\,t_{2m}|,\,\beta)
\,d\underline{t}
\end{align*}
for $\ovl{K}_f\,:=\,\sup_{s\in [0,1]}K_{f}(s,\,\beta)$.
We transform due to $s_i \,:=\, t_i\,-\,t_{i\,+\,1},\ i\,\leq\, 2m-1$ and $s_{2m}\,=\, t_{2m}$,
hence $-1\,\leq\, s_i\,\leq \,1,\ i\,=\,1,\ldots,\, 2m$, since integrating
a positive function we obtain
\begin{align*}
J(G,\,A,\,\beta)
&\leq  \Big(\int_{-1}^1\,K_{osc}(|s|,\,\beta)\,ds\Big)^m
\Big(\int_{-1}^1\,K_{f}(|s|,\,\beta)\,ds\Big)^{m\,-\,1}\\ \nonumber
& \cdot\sup_{s\in [0,1]}K_{f}(s,\,\beta).
\end{align*}
We recall that 
$$\int_{-1}^1K_{osc}(|s|,\,\beta)\,ds
\,=\,(2\Theta)^{-1}\int_{-1}^1\frac{\cosh(\beta \,\Theta \,|s|\,-\,\Theta\,\beta/2)}{\sinh(\Theta\,\beta/2)}\,ds
\,=\, 2(\Theta^2\, \beta)^{-1}$$
and 
\begin{align*}
\int_{-1}^1K_{f}(|s|,\,\beta)\,ds
\,&=\,\int_{-1}^1\int \frac{\cosh(\beta\,|s|\,|k|-\beta|k|/2)\,|f(k)|^2}{2\sinh(\beta\,|k|/2)}\,dk\,ds\\ 
\,&=\, 2\int \frac{|f(k)|^2}{\beta\,|k|}\,dk.
\end{align*}
Using $\coth(x)\leq 1+1/x$ and using convexity of $\cosh$, we obtain
\begin{equation*}
\sup_{s\in [0,\,1]}K_{f}(s,\,\beta)
\leq (1/2)\int |f(k)|^2\,dk \,+\, \frac{1}{\beta}\,\int \frac{|f(k)|^2}{|k|}\,dk.
\end{equation*}
Due to the fact, that $t\,\mapsto\, K_f(t,\,\beta)$ and $t\,\mapsto \,K_{osc}(t,\,\beta)$
attain their minima at $t\,=\,1/2$, we obtain the lower bound for
$J(\#A,\,\beta)$.
\end{proof}
%
%
\begin{Bem}\label{Bem3d}
In the literature there is one criterion for 
$\Omo\in \dom(e^{-\beta/2\,(\Lo\,+\,\Qg)})$, to our knowledge, that
can be applied in this situation \cite{DerezinskiJaksicPillet2003}.
One has to show that $\|e^{-\beta/2 \,\Qg}\,\Omo\|\,<\,\infty$. If
we consider the case, where the criterion holds for $\pm \lambda$, 
then the expansion in $\lambda$ converges,
\begin{align*}
\|e&^{-\beta/2\, \Qg}\,\Omo\|^2
\,=\,\sum_{n\,=\,0}^\infty \frac{(\lambda\,\beta)^{2n}}{(2n)!}\,\omel(q^{2n})\,\omf(\Phi(f)^{2n})\\ \nonumber
&=\sum_{n\,=\,0}^\infty \frac{(\lambda\,\beta)^{2n}}{(2n)!}\Big(\frac{(2n)!}{n!\,2^n}\Big)^2 
\,K_{osc}(0,\,\beta)^n\, K_{f}(0,\,\beta)^n\\ \nonumber
&=\sum_{n=0}^\infty (\lambda\,\beta)^{2n}\,\Theta^{-n}\,{2n\choose n}2^{-2n}
\Big(\coth(\Theta\, \beta/2)\int |f(k)|^2 \,\coth(\beta\,|k|/2)\,dk\Big)^n\\ \nonumber
&\,\geq\, \sum_{n\,=\,0}^\infty (\lambda\,\beta)^{2n}(4\,\Theta)^{-n}\,
\Big(\int |f(k)|^2\, dk \Big)^n.
\end{align*}
Obviously, for any value of $\lambda\,\not=\,0$, there is a $\beta\,>\,0$, 
for which  $\|e^{-\beta/2 \,\Qg}\,\Omo\|\,<\,\infty$ is not
fulfilled.
\end{Bem}
\bigskip
\textbf{Acknowledgments}:\\
This paper is part of the author's PhD requirements.
I am grateful to Volker Bach for many useful discussions
and helpful advice. The main part of this work 
was done during the author's stay at the Institut for Mathematics at
the University of Mainz. The work has been partially supported by the
DFG (SFB/TR 12).
\appendix
\section{}\label{App1}
\begin{Lemma}\label{Lem1App}
Let $f,\, g\,:\, \{ z\in \C\,:\, 0\,\leq\, \Re(z)\,\le\,\alpha\}\,\rightarrow \,\C$
continuous and analytic in the interior.
Moreover, assume that $f(t)\,=\,g(t)$ for $t\in \R$. Then $f\,=\,g$.
\end{Lemma}
\begin{proof}[Proof of \ref{Lem1App}]
Let $h\,:\, \{ z\in \C\,:\, |\,\Im(z)\,<\,\alpha\}\rightarrow \C$
defined by
\begin{equation}
h(z):= 
\begin{cases} f(\,z\,)\,-\, g(\,z\,),& \textrm{on }  \{ z\in \C\,:\, 0\,\leq\, \Im(z)\,<\,\alpha\}\\
\ovl{f(\,\ovl{z}\,)}- \ovl{g(\,\ovl{z}\,)},& \textrm{on } \{ z\in \C\,:\, -\alpha\,<\, \Im(z)\,<\,0\}
\end{cases}
\end{equation}
Thanks to the Schwarz reflection principle $h$ is analytic. 
Since $h(t)\,=\,0$ for all $t\in\R$, we get
$h\,=\,0$. Hence $f=g$ on $\{ z\in \C\,:\, 0\,\leq\, \Re(z)\,<\,\alpha\}$.
Since both $f$ and $g$ are continuous, we infer that $f=g$ on the whole domain.
\end{proof}
\begin{Lemma}\label{Lem2App}
Let $H$ be some self-adjoint operator in $\mathcal{H}$, $\alpha>0$ and $\phi\in \dom( e^{\alpha\,H} )$.
Then $\phi\in \dom( e^{z \,H})$ for $z\in \{ z\in \C\,:\, 0\,\leq\, \Re(z)\,\leq\,\alpha\}$.
$z\mapsto e^{z \,H}\phi$ is continuous on $\{ z\in \C\,:\, 0\,\leq\, \Re(z)\,\leq\,\alpha\}$
and analytic in the interior.
\end{Lemma}
\begin{proof}[Proof of \ref{Lem2App}]
Due to the spectral calculus we have
$$\int e^{2\Re\, z\, s} d\langle \phi\,|\,\mathbbm{E}_s\,\phi\rangle\le \int (1+ e^{2\alpha\, s}) d\langle \phi\,|\,\mathbbm{E}_s\,\phi\rangle=:C_1^2<\infty. $$
Thus $\phi\in \dom( e^{z \,H})$. Let $\psi\in \mathcal{H}$
and $f(z)=\langle \psi\,|\,e^{z \,H}\,\phi\rangle$. There 
is a sequence $\{\psi_n\}$ with $\psi_n\in \bigcup_{m\in \N} \operatorname{ran}\,\mathbbm{1}[ |H|\le m]$
and $\lim_{n\to\infty}\psi_n=\psi$. We set $f_n(z)= \langle \psi_n\,|\,e^{z\,H}\,\phi\rangle.$
It is not hard to see that $f_n$ is analytic, since $\psi_n$ is an analytic vector for $H$,
and that $|f_n(z)|\le C_1\,\|\psi_n\|$ and $\lim_{n\to \infty}f_n(z)=f(z)$. Thus $f$ is analytic and
hence $z\mapsto e^{z \,H}\phi$ is analytic.
Thanks to the dominated convergence theorem the right side of
\begin{equation}
\|e^{z_n H}\phi-e^{z H}\phi\|^2
\le \int (e^{2\Re z_n s}+e^{2\Re z s}- e^{\bar{z_n}s+zs}-e^{\bar{z}s+z_n s})
d\langle \phi\,|\,\mathbbm{E}_s\,\phi\rangle
\end{equation}
tends to zero for $\lim_{n\to\infty}z_n=z$. This implies the continuity of $z\mapsto e^{z \,H}\phi$.
\end{proof}
\begin{Lemma} \label{Lem0.2}
We have for $n_1+n_2\ge 1$%
\begin{equation}\label{eq3.52}
\int_{\Delta_1^{n}} C_\kappa(\unl{s})\, d\unl{s}_n
\,\leq\,  \frac{\const\,C^{n_2}}{(n_1+n_2)!\,(n+1)^{(1-2\gamma)\,n_2 -2} }
\end{equation}
\end{Lemma}
%
%
\begin{proof}[Proof of \ref{Lem0.2}]
We turn now to the integral
\begin{equation}\label{eq3.54}
\int_{\Delta_1^{n}} C_\kappa(\unl{s})\, d\unl{s}_{\,n}
\,=\,\int_{\Delta_1^{n}} (1\,-\,s_1\,+\,s_{n})^{-\alpha_1}
\prod_{i=1}^{n\,-\,1}(s_i\,-\,s_{i\,+\,1})^{-\alpha_i}\,d\underline{s}_n.
\end{equation}
We define for $k\,=\,1,\ldots,\,2n$, a change of coordinates by
$s_k\,=\, r_1\,-\,\sum_{j\,=\,2}^k r_j$, the integral transforms to
\begin{align}\label{eq3.55}
\int_{S^{n}} &(1\,-\,(r_2\,+\cdots+\,r_{n}))^{-\alpha_1}
\prod_{i\,=\,2}^{n}r_i^{-\alpha_i}\,	d\underline{r}_{\,n}\\ \nonumber
&= \int_{T^{n-1}} (1\,-\,(r_2\,+\cdots+\,r_{n}))^{1\,-\,\alpha_1}
\prod_{i=2}^{n}r_i^{-\alpha_i}\,d\underline{r}_{n-1}\\ \nonumber
\,&=\, \frac{\Gamma(1\,-\,\alpha_1)^{-1}\Gamma\big( 1\,-\,\gamma\big)^{2n_2}}
{\Gamma\big( n_1\,+\, 2n_2\,(1-\gamma)\big)}
\end{align}
where $ S^{2n} \,:=\,\{ \underline{r}\in \R^{2n}\, :\, 0\leq\, r_i\,\leq\, 1,\ r_2\,+\cdots +\,r_{2n}\,\leq\, r_1\}$
and $T^{2n-1} :=\{ \underline{r}\in\R^{2n-1}\, :\, 0\,\leq \,r_i\leq\, 1,\ r_2\,+\cdots +\,r_{2n} \,\leq\, 1\}$.
From the first to the second formula we integrate over $dr_1$.
The last equality follows from \cite[ Formula 4.635 (4)]{GradshteynRyzhik1980},
here $\Gamma$ denotes  the Gamma-function.\\
\indent From Stirling's formula we obtain 
\begin{equation}\label{eq3.56}
(2\pi)^{1/2}\, x^{x\,-\,1/2}\,e^{-x}\,\leq\, \Gamma(x)\,\leq\, (2\pi)^{1/2}\, x^{x\,-\,1/2}\,e^{-x\,+\,1},\quad x\,\geq\, 1.
\end{equation}
Since  $n_1\,+\,n_2\,\geq\, 1$ get
\begin{equation}\label{eq3.57}
\frac{\Gamma(n_1+n_2+1)}{\Gamma(n_1+2(1-\gamma)\,n_2)}
\,\leq\,  (n+1)^2
\Big(\frac{n_1+2(1-\gamma)n_2}{e}\Big)^{-(1-2\gamma)n_2}.
\end{equation}
Note that $\Gamma(n_1+n_2+1)=(n_1+n_2)!$.
\end{proof}
%
%
\begin{Lemma}\label{Lem-N1}
Let $(1\,+\,\alpha(k)^{-1/2})\,f_1,\ldots,(1\,
+\,\alpha(k)^{-1/2})\,f_{2m}\in \hph$ and $\sigma\in\{+,-\}^{2m}$.
Let $a^+\,=\,a^*$ and $a^-\,=\,a$
\begin{eqnarray*}\label{eq3.39}
\lefteqn{
\omf\big(a^{\sigma_{2m}}(e^{-\sigma_{2m} \,s_{2m}\,\alpha(k)}\,f_{2m})
\cdots a^{\sigma_1}(e^{-\sigma_1\, s_1\,\alpha(k)}\,f_1)\big)}\\ \nonumber
&=& \int f_{2m}^{\sigma_{2m}}(k_{2m},\,\tau_{2m})\cdots f_1^{\sigma_1}(k_{1},\,\tau_1)\,
\nu(d \unl{k}_{2m}\tensor d\unl{\tau}_{2m}),
\end{eqnarray*}
where $\nu(d \unl{k}_{2m}\tensor d\unl{\tau}_{2m}) $ is a 
measure on $(\R^3)^{2m}\times\{+,\,-\}^{2m}$ for phonons,
respectively on $(\R^3\times\{\pm\})^{2m}\times\{+,\,-\}^{2m}$
for photons,
and
\begin{equation}\label{eq3.43}
\nu(d\unl{k}_{2m}\tensor d\unl{\tau}_{2m})
\leq \sum_{P\in\mathcal{Z}_{2m}}\sum_{\unl{\tau}\in \{+,\,-\}^{2m}} \prod_{\{i\,>\,j\}\in P} 
\Big(\delta_{k_i,\,k_j}\coth(\beta\,\alpha(k_i)/2)\Big)\,d\unl{k}_{2m}.
\end{equation}
for  $f^+(k,\,\tau)\,:=\, f(k)\,\one[\tau\,=\,+]$ 
and $f^+(k,\,\tau)\,:=\,\ovl{f(k)}\,\one[\tau\,=\,-]$.
\end{Lemma}
\begin{proof}[Proof of \ref{Lem-N1}]
Since $\omf$ is quasi-free, we obtain with $a^+\,:=\,a^*$ and $a^-\,:=\,a$
\begin{eqnarray*}\label{eq3.36}
\lefteqn{
\omf\big(a^{\sigma_{2m}}(e^{-\sigma_{2m} \,s_{2m}\,\alpha(k)}\,f_{2m})
\cdots a^{\sigma_1}(e^{-\sigma_1 \,s_1\,\alpha(k)}\,f_1)\big)}\\ \nonumber
&=&\sum_{P\in \mathcal{Z}_2}\prod_{\stackrel{\{i,\,j\}\in P}{{i\,>\,j}}}
\omf\big(a^{\sigma_i}(e^{-\sigma_i\, s_i\, \alpha(k)}\,f_i)\,
a^{\sigma_j}(e^{-\sigma_j \,s_j\,\alpha(k)}\,f_j)\big),
\end{eqnarray*}
see Equation \eqref{QuasiFreeExp}. For the expectation of the so called two
point functions we obtain:
\begin{equation*}\label{EvalOneTwoPoinfOmf}
\omf\big(a^+(e^{ s_i\, \alpha(k)}\,f_i)\,a^{+}(e^{ s_j\,\alpha(k)}\,f_j)\big)\,=\,
0\,=\,\omf\big(a(e^{- s_i\,\alpha(k)}\,f_i)\,a(e^{- s_j\,\alpha(k)}\,f_j)\big),
\end{equation*}
such as
\begin{eqnarray*}\label{EvalOmf}
\omf\big(a^+(e^{ xs_i\alpha(k)}f_i)a^{-}(e^{- x\,s_j\alpha(k)}f_j)\big)
&\,=\,&\int f_i(k)\,\ovl{f_j(k)}\
\frac{e^{ x\,(s_i\,-\,s_j)\alpha(k)}}{e^{\beta \,\alpha(k)}-1}\,dk\\ \nonumber
\omf\big(a^-(e^{ xs_i\alpha(k)}f_i)\,a^{+}(e^{- xs_j\alpha(k)}f_j)\big)
&\,=\,&\int  f_j(k)\,\ovl{f_i(k)} \
\frac{e^{ (\beta+ xs_j-xs_i)\alpha(k)}}{e^{\beta\alpha(k)}-1}\,dk
\end{eqnarray*}
Hence it follows
\begin{eqnarray*}\label{QuasiFreeExp2}
\lefteqn{
\omf\big(a^{\sigma_{2m}}(e^{-\sigma_{2m}\, s_{2m}\,\alpha(k)}\,f_{2m})
\cdots a^{\sigma_1}(e^{-\sigma_1 \,s_1\,\alpha(k)}\,f_1)\big)}\\ \nonumber
&=& \int f_{2m}^{\sigma_{2m}}(k_{2m},\,\tau_{2m})\cdots f_1^{\sigma_1}(k_{1},\,\tau_1)\, 
\nu(d\unl{k}_{2m}\tensor d\unl{\tau}_{2m}),
\end{eqnarray*}
where $f^+(k,\,\tau)\,:=\, f(k)\,\one[\tau=+]$ and $f^-(k,\,\tau)\,:=\,\ovl{f(k)}\,\one[\tau\,=\,-]$.\\
$\nu(d^{3(2m)}k\tensor d^{2m}\tau) $ is a  measure on $(\R^3)^{2m}\times\{+,\,-\}^{2m}$,
which is defined by
\begin{align}\label{DefNu}
\sum_{P\in\mathcal{Z}_{2m}}&\ \sum_{\unl{\tau}\in \{+,\,-\}^{2m}}\, \prod_{\{i\,>\,j\}\in P} 
\delta_{\tau,\,-\tau}\,\delta_{k_i,\,k_j}\, \\ \nonumber
&
\Big(\delta_{\tau,\,+} 
\,\frac{e^{x\,(s_i\,-\,s_j)\,\alpha(k_i)}}{e^{\beta \,\alpha(k_i)}\,-\,1}\,
+\,\delta_{\tau,\,-}\, 
\frac{e^{(\beta\,-\,x\,(s_i\,-\,s_j))\,\alpha(k_i)}}{e^{\beta \,\alpha(k_i)}\,-\,1}\Big)\,d\unl{k}_{2m}.
\end{align}
\end{proof}
%

%

\Addresses
\end{document}